\documentclass[11pt]{article}

\usepackage[pdftex,paper=a4paper,left=3cm,right=3cm,top=4cm,bottom=4cm]{geometry}
\usepackage{amsmath, amssymb, amsthm}
\usepackage{xspace}
\usepackage{mathrsfs}
\usepackage{graphicx}
\usepackage{color}
\usepackage[T1]{fontenc}
\usepackage{lmodern}
\usepackage[pdftex,linktocpage,plainpages=false,breaklinks,hypertexnames=false]{hyperref}
\usepackage{algorithm}
\usepackage{algorithmic}
\usepackage{caption}
\usepackage{subcaption}
\usepackage{enumitem}

\newtheorem{theorem}{Theorem}
\newtheorem{lemma}[theorem]{Lemma}
\newtheorem{corollary}[theorem]{Corollary}
\newtheorem{definition}[theorem]{Definition}
\newtheorem{observation}[theorem]{Observation}

\newcommand{\RR}{\mathbb{R}}
\newcommand{\ZZ}{\mathbb{Z}}

\newcommand{\A}{A}
\newcommand{\B}{B}

\newcommand{\E}{E}
\newcommand{\eps}{\varepsilon}
\newcommand{\F}{\mathcal{F}}
\newcommand{\N}{\mathcal{N}}
\newcommand{\T}{{^\textnormal{T}}}
\renewcommand{\th}{^\text{th}}

\newcommand{\SET}[1]{\left\{#1\right\}}
\newcommand{\sSET}[1]{\{#1\}}
\newcommand{\DOT}{\,.}
\newcommand{\COMMA}{\,,}
\newcommand{\WHERE}{\,|\,}
\newcommand{\ceil}[1]{\left\lceil#1\right\rceil}

\newcommand{\DEF}{\mathbin{:=}}
\newcommand{\FED}{\mathbin{=:}}

\newcommand{\Ex}[2][]{\text{\rm\bf E}_{#1}\hspace{-0.03cm}\left[#2\right]}
\renewcommand{\Pr}[2][]{\text{\rm\bf Pr}_{#1}\hspace{-0.03cm}\left[#2\right]}

\newcommand{\rank}{\textnormal{rank}}
\newcommand{\ID}[1][]{\mathbb{I}_{#1}}
\newcommand{\ZERO}[1][]{\mathbb{O}_{#1}}
\newcommand{\NULL}[1][]{0_{#1}}

\newcommand{\junk}[1]{}

\newcommand{\SPAN}[1]{\textnormal{span}\{#1\}}

\providecommand{\institute}[1]{\date{#1}}
\providecommand{\email}[1]{\texttt{#1}}

\DeclareMathOperator*{\argmax}{arg\,max}

\newcommand{\diag}{\text{diag}}

\begin{document}

\title{Solving Totally Unimodular LPs with the\\ Shadow Vertex Algorithm\thanks{This research was supported by ERC Starting Grant 306465
(BeyondWorstCase).}}

\author{Tobias Brunsch \and Anna Gro{\ss}wendt \and Heiko R\"oglin}

\institute{
Department of Computer Science\\
University of Bonn,
Germany\\
\email{\small\{brunsch,grosswen,roeglin\}@cs.uni-bonn.de}}

\maketitle

\begin{abstract}
We show that the shadow vertex simplex algorithm can be used to solve linear programs in strongly polynomial time 
with respect to the number~$n$ of variables, the number~$m$ of constraints, and~$1/\delta$, where~$\delta$ is a
parameter that measures the flatness of the vertices of the polyhedron. This extends our recent result that the
shadow vertex algorithm finds paths of polynomial length (w.r.t.~$n$, $m$, and~$1/\delta$) between two given
vertices of a polyhedron~\cite{BrunschR13}. 

Our result also complements a recent result due to Eisenbrand and Vempala~\cite{GRE} who have shown
that a certain version of the random edge pivot rule solves linear programs with a running time that is
strongly polynomial in the number of variables~$n$ and~$1/\delta$, but independent of the number~$m$ of
constraints. Even though the running time of our algorithm depends on~$m$, it is significantly faster for
the important special case of totally unimodular linear programs, for which~$1/\delta\le n$ and
which have only~$O(n^2)$ constraints.
\end{abstract}

\section{Introduction}

The shadow vertex algorithm is a well-known pivoting rule for the simplex method that has gained
attention in recent years because it was shown to have polynomial running time in the model
of smoothed analysis~\cite{SpielmanTeng}. Recently we have observed that it can also be used
to find short paths between given vertices of a polyhedron~\cite{BrunschR13}. Here short
means that the path length is~$O(\frac{mn^2}{\delta^2})$, where~$n$ denotes the number of
variables, $m$ denotes the number of constraints, and~$\delta$ is a parameter of the polyhedron
that we will define shortly. 

Our result left open the question whether or not it is also possible to solve linear programs in
polynomial time with respect to~$n$,~$m$, and~$1/\delta$ by the shadow vertex simplex algorithm.  
In this article we resolve this question and introduce a variant of the shadow vertex simplex
algorithm that solves linear programs in strongly polynomial time with respect to these
parameters. 

For a given matrix~$A = [a_1, \ldots, a_m]\T \in \RR^{m \times n}$ and vectors~$b \in \RR^m$ and~$c_0\in \RR^n$ our
goal is to solve the linear program~$\max \lbrace c_0\T x \WHERE Ax \leq b \rbrace$.
We assume without loss of generality that~$\|c_0\|=1$ and~$\|a_i\|=1$ for every row~$a_i$ of the constraint matrix.
 
\begin{definition}\label{def:delta}
The matrix~$A$ satisfies the
\emph{$\delta$-distance property} if the following condition holds:
For any~$I\subseteq\{1,\ldots,m\}$ and any~$j\in\{1,\ldots,m\}$,
if $a_j\notin\SPAN{a_i \WHERE i\in I}$ then~$\mathrm{dist}(a_j,\SPAN{a_i \WHERE i\in I})\ge\delta$.
In other words, if~$a_j$ does not lie in the subspace spanned by the~$a_i$, $i\in I$, then its
distance to this subspace is at least~$\delta$.
\end{definition}

We present a variant of the shadow vertex simplex algorithm that solves linear programs in
strongly polynomial time with respect to~$n$, $m$, and~$1/\delta$, where~$\delta$ denotes the 
largest~$\delta'$ for which the constraint matrix of the linear program satisfies the $\delta'$-distance property.
(In the following theorems,
we assume~$m\ge n$. If this is not the case, we use the method from Section~\ref{dimension} to add
irrelevant constraints so that~$A$ has rank~$n$. Hence, for instances that have fewer constraints
than variables, the parameter~$m$ should be replaced by~$n$ in all bounds.)

\begin{theorem}\label{thm:MainNumberOfPivots}
There exists a randomized variant of the shadow vertex simplex algorithm (described in Section~\ref{fullalgorithm}) that solves linear programs
with $n$~variables and $m$~constraints satisfying the $\delta$-distance property 
using~$O\big(\frac{mn^3}{\delta^2}\cdot \log\big(\frac{1}{\delta}\big)\big)$ pivots
in expectation if a basic feasible solution is given.
A basic feasible solution can be found using~$O\big(\frac{m^5}{\delta^2}\cdot \log\big(\frac{1}{\delta}\big)\big)$ pivots in expectation.
\end{theorem}
We stress that the algorithm can be implemented without knowing the parameter~$\delta$.
From the theorem it follows that the running time of the algorithm is strongly polynomial with respect to the number~$n$ of variables, the number~$m$ of constraints, and~$1/\delta$
because every pivot can be performed in time $O(mn)$ in the arithmetic model of computation (see Section~\ref{runtime}).\footnote{By \emph{strongly polynomial with respect to~$n$, $m$, and~$1/\delta$} 
we mean that the number of steps in the arithmetic model of computation is bounded polynomially in~$n$, $m$, and~$1/\delta$ and the size of the numbers occurring during the algorithm is polynomially bounded
in the encoding size of the input.}

Let~$A\in\ZZ^{m\times n}$ be an integer matrix and let~$A'\in\RR^{m\times n}$ be the matrix
that arises from~$A$ by scaling each row such that its norm equals~$1$. If~$\Delta$ denotes
an upper bound for the absolute value of any sub-determinant of~$A$, then~$A'$ satisfies the
$\delta$-distance property for~$\delta = 1/(\Delta^2 n)$~\cite{BrunschR13}. For such matrices~$A$
Phase~1 of the simplex method can be implemented more efficiently and we obtain the following result. 

\begin{theorem}\label{thm:MainNumberOfPivots2}
For integer matrices~$A\in\ZZ^{m\times n}$,
there exists a randomized variant of the shadow vertex simplex algorithm (described in Section~\ref{fullalgorithm}) that solves linear programs
with $n$~variables and $m$~constraints 
using~$O\big(mn^5\Delta^4 \log(\Delta+1)\big)$ pivots in expectation if a basic feasible solution is given,
where~$\Delta$ denotes an upper bound for the absolute value of any sub-determinant of~$A$.
A basic feasible solution can be found using~$O\big(m^6\Delta^4 \log(\Delta+1)\big)$ pivots in expectation.
\end{theorem}

Theorem~\ref{thm:MainNumberOfPivots2} implies in particular that totally unimodular linear programs
can be solved by our algorithm with~$O\big(mn^5\big)$ pivots in expectation if a basic feasible solution is given
and with~$O\big(m^6 \big)$ pivots in expectation otherwise.

Besides totally unimodular matrices there are also other classes of matrices for which~$1/\delta$ is polynomially
bounded in~$n$. Eisenbrand and Vempala~\cite{GRE} observed, for example, that~$\delta=\Omega(1/\sqrt{n})$
for edge-node incidence matrices of undirected graphs with~$n$ vertices. One can also argue
that~$\delta$ can be interpreted as a condition number of the matrix~$A$ in the following sense: If~$1/\delta$ is large
then there must be an $(n\times n)$-submatrix of~$A$ of rank~$n$ that is almost singular.

\subsection{Related Work}

\paragraph{Shadow vertex simplex algorithm}
We will briefly explain the geometric intuition behind the shadow vertex simplex algorithm.
For a complete and more formal description, we refer the reader to~\cite{Borgwardt86} or~\cite{SpielmanTeng}. 
Let us consider the linear program~$\max \lbrace c_0\T x \WHERE Ax \leq b \rbrace$ and let~$P= \SET{ x \in \RR^n \WHERE Ax \leq b }$
denote the polyhedron of feasible solutions. Assume that an initial vertex~$x_1$ of~$P$ is known and assume, 
for the sake of simplicity, that there is a unique optimal vertex~$x^{\star}$ of~$P$ that maximizes the objective
function~$c_0\T x$. The shadow vertex pivot rule first computes a vector~$w\in\RR^n$ such that the vertex~$x_1$ 
minimizes the objective function~$w\T x$ subject to~$x\in P$. Again for the sake of simplicity, let
us assume that the vectors~$c_0$ and~$w$ are linearly independent.

In the second step, the polyhedron~$P$ is projected onto the plane spanned by the vectors~$c_0$ and~$w$.
The resulting projection is a (possibly open) polygon~$P'$ and one can show that the projections 
of both the initial vertex~$x_1$ and the optimal vertex~$x^{\star}$ are vertices of this polygon.
Additionally, every edge between two vertices~$x$ and~$y$ of~$P'$ corresponds to an edge of~$P$ between
two vertices that are projected onto~$x$ and~$y$, respectively. Due to these properties a path from the projection of~$x_1$ to
the projection of~$x^{\star}$ along the edges of~$P'$ corresponds to a path from~$x_1$ to~$x^{\star}$ 
along the edges of~$P$.

This way, the problem of finding a path from~$x_1$ to~$x^{\star}$ on the polyhedron~$P$ is reduced to
finding a path between two vertices of a polygon. There are at most two such paths and the shadow vertex
pivot rule chooses the one along which the objective~$c_0\T x$ improves.

\paragraph{Finding short paths}
In~\cite{BrunschR13} we considered the problem of finding a short path between two given vertices~$x_1$ and~$x_2$ 
of the polyhedron~$P$ along the edges of~$P$. Our algorithm is the following variant of the shadow vertex algorithm: 
Choose two vectors~$w_1,w_2\in\RR^n$ such that~$x_1$ uniquely minimizes~$w_1\T x$ subject to~$x\in P$
and~$x_2$ uniquely maximizes~$w_2\T x$ subject to~$x\in P$. Then project the polyhedron~$P$ onto the plane
spanned by~$w_1$ and~$w_2$ in order to obtain a polygon~$P'$. Let us call the projection~$\pi$. By the same
arguments as above, it follows that~$\pi(x_1)$ and~$\pi(x_2)$ are vertices of~$P'$
and that a path from~$\pi(x_1)$ to~$\pi(x_2)$ along the edges of~$P'$ can be translated into a path
from~$x_1$ to~$x_2$ along the edges of~$P$. Hence, it suffices to compute such a path to solve the problem.
Again computing such a path is easy because~$P'$ is a two-dimensional polygon. 

The vectors~$w_1$ and~$w_2$ are not uniquely determined, but they can be chosen from cones that are determined
by the vertices~$x_1$ and~$x_2$ and the polyhedron~$P$. We proved in~\cite{BrunschR13} that the expected path length is~$O(\frac{mn^2}{\delta^2})$
if~$w_1$ and~$w_2$ are chosen randomly from these cones. For totally unimodular matrices this implies
that the diameter of the polyhedron is bounded by~$O(mn^4)$, which improved a previous result
by Dyer and Frieze~\cite{DyerF94} who showed that for this special case paths of length~$O(m^3n^{16}\log(mn))$
can be computed efficiently.

Additionally, Bonifas et al.~\cite{BonifasDEHN12} proved that in a polyhedron defined by an integer matrix~$A$ between any pair of vertices there exists
a path of length~$O(\Delta^2 n^4 \log (n\Delta))$ where~$\Delta$ is the largest absolute value of any
sub-determinant of~$A$. 
For the special case that~$A$ is a totally unimodular matrix, this bound simplifies to~$O(n^4\log{n})$.
Their proof is non-constructive, however.

\paragraph{Geometric random edge}
Eisenbrand and Vempala~\cite{GRE} have presented an algorithm that solves a linear program~$\max \lbrace c_0\T x | Ax \leq b \rbrace$
in strongly polynomial time with respect to the parameters~$n$ and~$1/\delta$. Remarkably the running time of their algorithm does
not depend on the number~$m$ of constraints. Their algorithm is based on a variant of the random edge pivoting rule.
The algorithm performs a random walk on the vertices of the polyhedron whose transition probabilities 
are chosen such that it quickly attains a distribution close to its stationary distribution.
 
In the stationary distribution the random walk is likely at a vertex~$x_c$ that optimizes an objective function~$c\T x$
with~$\|c_0-c\|<\frac{\delta}{2n}$. The $\delta$-distance property guarantees that~$x_c$ and the optimal vertex~$x^{\star}$
with respect to the objective function~$c_0\T x$ lie on a common facet. This facet is then identified and the algorithm
is run again in one dimension lower. This is repeated at most~$n$ times until all facets of the optimal vertex~$x^{\star}$
are identified. The number of pivots to identify one facet of~$x^{\star}$ is proven to be~$O(n^{10}/\delta^{8})$.
A single pivot can be performed in polynomial time but determining the right transition probabilities
is rather sophisticated and requires to approximately integrate a certain function over a convex body.

Let us point out that the number of pivots of our algorithm depends on the number~$m$ of constraints.
However, Heller showed that for the important special case of totally unimodular linear programs~$m=O(n^2)$~\cite{Heller57}.
Using this observation we also obtain a bound that depends polynomially only on~$n$ for totally
unimodular matrices.

\paragraph{Combinatorial linear programs}
{\'{E}}va Tardos has proved in 1986 that combinatorial linear programs can be solved in strongly polynomial time~\cite{Tardos1986}.
Here combinatorial means that~$A$ is an integer matrix whose largest entry is polynomially bounded in~$n$.
Her result implies in particular that totally unimodular linear programs can be solved in strongly polynomial
time, which is also implied by Theorem~\ref{thm:MainNumberOfPivots2}. However, the proof and the techniques used to
prove Theorem~\ref{thm:MainNumberOfPivots2} are completely different from those in~\cite{Tardos1986}.

\subsection{Our Contribution}\label{sub:contribution}

We replace the random walk in the algorithm of Eisenbrand and Vempala by the shadow vertex algorithm.
Given a vertex~$x_0$ of the polyhedron~$P$ we choose an objective function~$w\T x$ for which~$x_0$ is
an optimal solution. As in~\cite{BrunschR13} we choose~$w$ uniformly at random from the cone determined 
by~$x_0$. Then we randomly perturb each coefficient in the given objective function~$c_0\T x$ by a small amount.
We denote by~$c\T x$ the perturbed objective function.
As in~\cite{BrunschR13} we prove that the projection of the polyhedron~$P$ onto the plane spanned by~$w$ and~$c$
has~$O\big(\frac{mn^2}{\delta^2}\big)$ edges in expectation. If the perturbation is so small that~$\|c_0-c\|<\frac{\delta}{2n}$,
then the shadow vertex algorithm yields with~$O\big(\frac{mn^2}{\delta^2}\big)$ pivots a solution that has
a common facet with the optimal solution~$x^{\star}$. We follow the same approach as Eisenbrand and Vempala
and identify the facets of~$x^{\star}$ one by one with at most~$n$ calls of the shadow vertex algorithm.

The analysis in~\cite{BrunschR13} exploits that the two objective functions possess the same type of randomness
(both are chosen uniformly at random from some cones). This is not the case anymore because every component
of~$c$ is chosen independently uniformly at random from some interval. This changes the analysis significantly
and introduces technical difficulties that we address in this article. 

The problem when running the simplex method is that a feasible solution needs to be
given upfront. Usually, such a solution is determined in Phase~1 by solving a modified linear program
with a constraint matrix~$A'$ for which a feasible solution is known and whose optimal solution is feasible for the linear program
one actually wants to solve. There are several common constructions for this modified linear program,
it is, however, not clear how the parameter~$\delta$ is affected by modifying the linear program. 
To solve this problem, Eisenbrand and Vempala~\cite{GRE} have suggested a method for Phase~1 for which the
modified constraint matrix~$A'$ satisfies the $\delta$-distance property for the same~$\delta$ as the matrix~$A$.
However, their method is very different from usual textbook methods and needs to solve~$m$ different linear programs
to find an initial feasible solution for the given linear program.
We show that also one of the usual textbook methods can be applied.
We argue that $1/\delta$ increases by a factor of at most~$\sqrt{m}$ and that~$\Delta$, 
the absolute value of any sub-determinant of~$A$,
does not change at all in case one considers integer matrices. In this construction, 
the number of variables increases from~$n$ to~$n+m$.

\subsection{Outline and Notation}

In the following we assume that we are given a linear program $\max \lbrace c_0\T x \WHERE Ax \leq b \rbrace$
with vectors $b \in \RR^m$ and $c_0 \in \RR^n$ and a matrix $A=[a_1,\dots,a_m]\T \in \RR^{m \times n}$. 
Moreover, we assume that $\|c_0\|=\|a_i\|=1$ for all $ i \in [m]$, where $[m]:=\{1,\ldots,m\}$
and $\|\cdot\|$ denotes the Euclidean norm. This entails no loss of generality since any linear
program can be brought into this form by scaling the objective function and the constraints
appropriately.
For a vector $x \in \RR^n\setminus\{0^n\}$ we denote by $\N(x) = \frac{1}{\|x\|} \cdot x$ the normalization of vector~$x$. 

For a vertex $v$ of the polyhedron $P=\SET{ x \in \RR^n \WHERE Ax \leq b }$ we call the set of row indices $B_v=\lbrace i\in\{1,\ldots,m\} \WHERE 
a_i\cdot v=b_i \rbrace$ \textit{basis} of $v$. Then the \textit{normal cone}
$C_v$ of $v$ is given by the set 
\[
   C_v=\SET{ \sum \limits_{i \in B_v} \lambda_i a_i\WHERE \lambda_i\ge 0 }.
\]

We will describe our algorithm in Section~\ref{sec:Algorithm} where we assume that the linear
program in non-degenerate, that~$A$ has full rank~$n$, and that the polyhedron~$P$
is bounded. We have already described in Section~3 of~\cite{BrunschR13} that the linear program
can be made non-degenerate by slightly perturbing the vector~$b$. This does not affect the
parameter~$\delta$ because~$\delta$ depends only on the matrix~$A$. In Appendix~\ref{specialcases}
we discuss why we can assume that~$A$ has full rank and why~$P$ is bounded. There are, of course,
textbook methods to transform a linear program into this form. However, we need to be careful that
this transformation does not change~$\delta$.  

In Section~\ref{sec:analysis} we analyze our algorithm and prove Theorem~\ref{thm:MainNumberOfPivots}. 
In Section~\ref{initial} we discuss how Phase~1 of the simplex method can be implemented
and in Appendix~\ref{sdelta} we give an alternative definition of~$\delta$ and discuss some
properties of this parameter.

\section{Algorithm}\label{fullalgorithm}

Given a linear program $\max \lbrace c_0 \T x \WHERE Ax \leq b \rbrace$ and a basic feasible 
solution~$x_0$, our algorithm randomly perturbs each coefficient of the vector~$c_0$ by
at most~$1/\phi$ for some parameter~$\phi$ to be determined later. Let us call the resulting
vector~$c$. The next step is then to use the shadow vertex algorithm to compute a
path from $x_0$ to a vertex $x_c$ which maximizes the function $c\T x$ for $x \in P$.
For~$\phi>\frac{2n^{3/2}}{\delta}$ one can argue that the solution~$x$ has a facet in common with the 
optimal solution~$x^{\star}$ of the given linear program with objective function~$c_0\T x$. Then the
algorithm is run again on this facet one dimension lower until all facets that define~$x^{\star}$
are identified.

This section is organized as follows. In Section~\ref{reduct} we repeat a construction
from~\cite{GRE} to project a facet of the polyhedron~$P$ into the space $\RR^{n-1}$ without
changing the parameter~$\delta$. This is crucial for being able to identify the facets that
define~$x^{\star}$ one after another. In Section~\ref{identify} we also repeat an argument
from~\cite{GRE} that shows how a common facet of~$x_c$ and~$x^{\star}$ can be identified if~$x_c$
is given. Section~\ref{sec:Algorithm} presents the shadow vertex algorithm, the main building block
of our algorithm. Finally in Section~\ref{runtime} we discuss the running time of a single pivot step   
of the shadow vertex algorithm.

\subsection{Reducing the Dimension} \label{reduct}

Assume that we have identified an element~$a_i$, $i \in [m]$, of the optimal basis~$x^{\star}$
(i.e., $a_ix^{\star}=b_i$). In~\cite{GRE} it is described how to reduce in this case
the dimension of the linear program by one without changing the parameter~$\delta$.
We repeat the details. Without loss of generality we may assume that~$a_1$ is an element 
of the optimal basis. Let~$Q \in \mathbb R^{n \times n}$ be an orthogonal matrix that 
rotates~$a_1$ into the first unit vector~$e_1$. Then the following linear programs are equivalent:
\begin{equation}\label{eqn:LP}
    \max \lbrace c_0^Tx\WHERE x \in \mathbb R^n, Ax \leq b \rbrace 
\end{equation}
and
\[
    \max \lbrace c_0^TQx\WHERE x \in \mathbb R^n, AQx \leq b \rbrace.
\]
In the latter linear program the first constraint is of the form~$x_1 \leq b_1$. 
We set this constraint to equality and subtract this equation from the other constraints
(i.e., we project each row into the orthogonal complement of~$e_1$). 
Thus, we end up with a linear program of dimension~$n-1$.
Lemma~\ref{deltaproperties} shows that the $\delta$-distance does not change under multiplication
with an orthogonal matrix. Furthermore, Lemma~3 of~\cite{GRE} ensures that the $\delta$-distance
property is not destroyed by the projection onto the orthogonal complement.

\subsection{Identifying an Element of the Optimal Basis}\label{identify}

In this section we repeat how an element of the optimal basis can be identified if an optimal
solution~$x_c$ for an objective function~$c\T x$ with~$\|c_0-c\|<\delta/(2n)$ is given (see
also~\cite{GRE}). 

\begin{lemma}[Lemma 2 of~\cite{GRE}]\label{lemma:IdentificationMain}
Let $B \subseteq \lbrace 1, \dots, m \rbrace$ be the optimal basis of the linear program~\eqref{eqn:LP} 
and let~$B'$ be an optimal basis of the linear program~\eqref{eqn:LP} with~$c_0$ being replaced by~$c$, where~$\| c_0-c \| < \delta/(2n)$ holds. Consider the conic combination 
\begin{align*}
c=\sum \limits_{j \in B'} \mu_j a_j.
\end{align*}
For $k \in B' \setminus B$, one has $\| c_0-c \| \geq \delta \cdot \mu_k$.
\end{lemma}

The following corollary whose proof can also be found in~\cite{GRE}
gives a constructive way to identify an element of the optimal basis.

\begin{corollary}\label{gre-cor}
Let~$c\in\RR^{n}$ be such that $\| c_0-c \| < \delta/(2 \cdot n)$ and let~$\mu_j$, $B$, and~$B'$ be defined
as in Lemma~\ref{lemma:IdentificationMain}.
There exists at least one coefficient~$\mu_k$ with~$\mu_k > 1/n \cdot (1- \delta /(2\cdot n))$ and
any~$k$ with this property is an element of the optimal basis~$B$ (assuming~$\|c_0\|=1$).
\end{corollary}

The corollary implies that given a solution~$x_c$ that is optimal for an objective function~$c\T x$
with~$\|c_0-c\|<\delta/(2n)$, one can identify an element of the optimal basis by solving 
the system of linear equations  
\begin{align*}
[a_1',\dots , a_n'] \cdot \mu = c,
\end{align*}
where the~$a_i'$ denote the constraints that are tight in~$x_c$.

\subsection{The Shadow Vertex Method}\label{sec:Algorithm}

In this section we assume that we are given a linear program of the form $\max \lbrace c_0\T x \WHERE x \in P \rbrace$, where $P = \lbrace x \in \RR^n \WHERE Ax \leq b \rbrace$ is a bounded polyhedron (i.e., a polytope), and
a basic feasible solution $x_0 \in P$. We assume $\|c_0\|=\|a_i\|=1$ for all rows~$a_i$ of~$A$.
Furthermore, we assume that the linear program is non-degenerate.

Due to the assumption~$\|c_0\|=1$ it holds~$c_0\in[-1,1]^n$.
Our algorithm slightly perturbs the given objective function~$c_0\T x$ at random. For each component~$(c_0)_i$ of~$c_0$ it chooses
an arbitrary interval~$I_i\subseteq[-1,1]$ of length~$1/\phi$ with~$(c_0)_i\in I_i$, where~$\phi$ denotes a parameter that will be given to the algorithm.
Then a random vector~$c\in[-1,1]^n$ is drawn as follows: Each component~$c_i$ of~$c$ is chosen independently uniformly
at random from the interval~$I_i$. We denote the resulting random vector by~$\text{pert}(c_0,\phi)$.
Note that we can bound the norm of the difference $\|c_0-c\|$ between the vectors $c_0$ and $c$ from above by $\frac{\sqrt{n}}{\phi}$.

The shadow vertex algorithm is given as Algorithm~\ref{algorithm:SV}. It is assumed that~$\phi$ is given to the algorithm as 
a parameter. We will discuss later how we can run the algorithm without knowing this parameter. Let us remark that the
Steps~\ref{line:shadow vertex polytope} and~\ref{line:walk} in Algorithm~\ref{algorithm:SV} are actually not executed separately.
Instead of computing the whole projection~$P'$ in advance, the edges of~$P'$ are computed on the fly one after another. 

\begin{algorithm*}[th]
  \caption{Shadow Vertex Algorithm}
  \label{algorithm:SV}
  \begin{algorithmic}[1]
  
    \STATE Generate a random perturbation~$c=\text{pert}(c_0,\phi)$ of~$c_0$.\label{line:randomnessc}

    \STATE Determine~$n$ linearly independent rows $u_k\T$ of~$A$ for which $u_k\T x_0 = b_k$. \label{line:opt vectors}

    \STATE Draw a vector $\lambda \in (0, 1]^n$ uniformly at random. \label{line:randomness}

    \STATE Set $w = -\left[ u_1, \ldots, u_n \right] \cdot \lambda$. \label{line:opt direction}

    \STATE Use the function $\pi: x \mapsto \big( c\T x, w\T x \big)$ to project~$P$ onto the Euclidean plane and obtain the shadow vertex polygon $P' = \pi(P)$. \label{line:shadow vertex polytope}

    \STATE Walk from~$\pi(x_0)$ along the edges of~$P'$ in increasing direction of the first coordinate until a rightmost vertex~$\tilde{x}_c$ of~$P'$ is found. \label{line:walk}

    \STATE Output the vertex~$x_c$ of~$P$ that is projected onto~$\tilde{x}_c$.

  \end{algorithmic}
\end{algorithm*}
Note that
\[
  \|w\|
  \leq \sum_{k=1}^n \lambda_k \cdot \|u_k\|
  \leq \sum_{k=1}^n \lambda_k
  \leq n,
\]
where the second inequality follows because all rows of~$A$ are assumed to have norm~1.

The Shadow Vertex Algorithm yields a path from the vertex~$x_0$ to a vertex~$x_c$ 
that is optimal for the linear program $\max \lbrace c\T x \WHERE x \in P \rbrace$ 
where $P = \lbrace x \in \RR^n \WHERE Ax \leq b \rbrace$. The following theorem (whose proof can be found in Section~\ref{sec:analysis}) bounds the expected length of this path, i.e., the number
of pivots.

\begin{theorem}\label{main}
For any~$\phi\ge \sqrt{n}$ the expected number of edges on the path output by Algorithm~\ref{algorithm:SV} is
$O\Big(\frac{mn^2}{\delta^2}+\frac{m\sqrt{n}\phi}{\delta}\Big)$.
\end{theorem}

Since~$\|c_0-c\| \le \frac{\sqrt{n}}{\phi}$ choosing~$\phi>\frac{2n^{3/2}}{\delta}$
suffices to ensure~$\|c_0-c\| < \frac{\delta}{2n}$.
Hence, for such a choice of~$\phi$, by Corollary~\ref{gre-cor}, the vertex $x_c$ has a facet in common with the optimal solution of the linear program $\max \lbrace c_0\T x \WHERE x \in P \rbrace$ 
and we can reduce the dimension of the linear program as discussed in 
Section~\ref{reduct}. This step is repeated at most $n$ times. It is important that we can start each repetition with a known feasible
solution because the transformation in Section~\ref{reduct} maps the optimal solution of the linear program of repetition~$i$ onto a feasible solution with which 
repetition~$i+1$ can be initialized.
Together with Theorem~\ref{main} this implies that an optimal solution of the linear program~\eqref{eqn:LP} 
can be found by performing in expectation $O\Big(\frac{mn^3}{\delta^2}+\frac{mn^{3/2}\phi}{\delta}\Big)$ pivots
if a basic feasible solution~$x_0$ and the right choice of~$\phi$ are given.
We will refer to this algorithm as \emph{repeated shadow vertex algorithm}.

Since~$\delta$ is not known to the algorithm, the right choice for~$\phi$ cannot easily be computed. Instead we will
try values for~$\phi$ until an optimal solution is found. For~$i\in\mathbb{N}$ let~$\phi_i=2^in^{3/2}$.
First we run the repeated shadow vertex algorithm with~$\phi=\phi_0$ and check whether the returned solution
is an optimal solution for the linear program $\max \lbrace c_0\T x \WHERE x \in P \rbrace$.
If this is not the case, we run the repeated shadow vertex algorithm with~$\phi=\phi_1$,
and so on. We continue until an optimal solution is found.
For~$\phi=\phi_{i^{\star}}$ with~$i^{\star}=\big\lceil\log_2\big(1/\delta\big)\big\rceil+2$ this
is the case because~$\phi_{i^{\star}}>\frac{2n^{3/2}}{\delta}$.

Since~$\phi_{i^{\star}}\le \frac{8n^{3/2}}{\delta}$, in accordance with Theorem~\ref{main},
each of the at most~$i^{\star}=O(\log(1/\delta))$ calls of the repeated shadow vertex algorithm uses
in expectation
\[
      O\bigg(\frac{mn^3}{\delta^2}+\frac{mn^{3/2}\phi_{i^{\star}}}{\delta}\bigg)
    = O\bigg(\frac{mn^3}{\delta^2}\bigg).
\]
pivots. Together this proves the first part of Theorem~\ref{thm:MainNumberOfPivots}. 
The second part follows with Lemma~\ref{lemma:Phase I delta}, which states that Phase~1 can be realized
with increasing~$1/\delta$ by at most~$\sqrt{m}$ and increasing the number of variables from~$n$ to~$n+m\le 2m$. 
This implies that the expected number of
pivots of each call of the repeated shadow vertex algorithm in Phase~1 is~$O(m(n+m)^3\sqrt{m}^2/\delta^2) = O(m^5/\delta^2)$. 
Since~$1/\delta$ can increase by a factor of~$\sqrt{m}$, the argument above yields that we need to run the repeated
shadow vertex algorithm at most~$i^{\star}=O(\log(\sqrt{m}/\delta))$ times in Phase~1 to find a basic feasible solution.  
By setting~$\phi_i=2^i\sqrt{m}(n+m)^{3/2}$ instead of~$\phi_i=2^i(n+m)^{3/2}$ this number can be
reduced to~$i^{\star}=O(\log(1/\delta))$ again.

Theorem~\ref{thm:MainNumberOfPivots2} follows from Theorem~\ref{thm:MainNumberOfPivots} using the following fact from~\cite{BrunschR13}:
Let~$A\in\ZZ^{m\times n}$ be an integer matrix and let~$A'\in\RR^{m\times n}$ be the matrix
that arises from~$A$ by scaling each row such that its norm equals~$1$. If~$\Delta$ denotes
an upper bound for the absolute value of any sub-determinant of~$A$, then~$A'$ satisfies the
$\delta$-distance property for~$\delta = 1/(\Delta^2 n)$. Additionally Lemma~\ref{lemma:Phase I Delta}
states that Phase~1 can be realized without increasing~$\Delta$ but with increasing the number of variables from~$n$ to~$n+m\le 2m$.
Substituting~$1/\delta=\Delta^2n$ in Theorem~\ref{thm:MainNumberOfPivots} almost yields Theorem~\ref{thm:MainNumberOfPivots2}
except for a factor~$O(\log(\Delta^2 n))$ instead of~$O(\log(\Delta+1))$. This factor results from the number~$i^{\star}$ of
calls of the repeated shadow vertex algorithm. The desired factor of~$O(\log(\Delta+1))$ can be achieved by
setting~$\phi_i=2^in^{5/2}$ if a basic feasible solution is known and~$\phi_i=2^i(n+m)^{5/2}$ in Phase~1.

\subsection{Running Time}\label{runtime}

So far we have only discussed the number of pivots. Let us now calculate 
the actual running time of our algorithm.
For an initial basic feasible solution~$x_0$ the repeated shadow vertex algorithm repeats the
following three steps until an optimal solution is found. Initially let~$P'=P$.
\renewcommand{\leftmargini}{2cm}
\begin{enumerate}

\item[\textbf{Step 1:}] Run the shadow vertex algorithm for the linear program $\max \lbrace c\T x \WHERE x \in P' \rbrace$, where $c=\text{pert}(c_0,\phi)$.
                        We will denote this linear program by~$LP'$. 
\item[\textbf{Step 2:}] Let~$x_c$ denote the returned vertex in Step 1, which is optimal for the objective function $c\T x$. Identify an element $a'_i$ of~$x_c$ that is in common with the optimal basis.
\item[\textbf{Step 3:}] Calculate an orthogonal matrix $Q \in \RR^{n \times n}$ that rotates $a'_i$ into the first unit vector $e_1$ as described in Section~\ref{reduct} and set $LP'$ to the projection of the current~$LP'$ onto the orthogonal complement. Let~$P'$ denote the polyhedron of feasible solutions of~$LP'$.
\end{enumerate}
First note that the three steps are repeated at most~$n$ times during the algorithm.
In Step 1 the shadow vertex algorithm is run once. Step 1 to Step 4 of Algorithm~\ref{algorithm:SV} can be performed in time~$O(m)$ as we assumed $P$ to be non-degenerate (this implies $P'$ to be non-degenerate in each further step). Step 5 and Step 6 can be implemented with strongly polynomial running time in a tableau form, described in \cite{Borgwardt86}.
The tableau can be set up in time $O((m-d)d^3)= O(mn^3)$ where $d$ is the dimension of $P'$. By Theorem 1 of \cite{Borgwardt86} we can identify for a vertex on a path the row which leaves the basis and the row which is added to the basis in order to move to the next vertex in time $O(m)$ using the tableau. After that, the tableau has to be updated. This can be done in $O((m-d)d)= O(mn)$ steps. 
Using this and Theorem~\ref{main} we can compute the path from $x_0$ to $x_c$ in expected
time $O\big(mn^3+mn\cdot \big(\frac{mn^2}{\delta^2}+\frac{m\sqrt{n}\phi}{\delta}\big)\big) = O\big(\frac{m^2n^3}{\delta^2}+\frac{m^2n^{3/2}\phi}{\delta}\big)$.
Using that~$\phi \le \frac{8n^{3/2}}{\delta}$, as discussed above, yields a running time
of $O\big(\frac{m^2n^3}{\delta^2}\big)$.
 
Once we have calculated the basis of $x_c$ we can easily compute the element $a_{i}$ of the basis that is also an element of the optimal basis. Assume the rows $a'_1, \dots, a'_n$ are the basis of $x_c$. As mentioned in Section~\ref{identify} we can solve the system of linear equations $[a'_1, \dots, a'_n]\mu=c$ and choose the row for which the coefficient $\mu_i$ is maximal. Then $a'_i$ is part of the optimal basis. As a consequence, Step 2 can be performed in time $O(n^3)$. Moreover solving a system of linear equations is possible in strongly polynomial time using Gaussian elimination.

In Step 3, we compute an orthogonal matrix $Q \in \RR^{d \times d}$ such that $e_1 Q=a_{i}$. Since~$Q$ is orthogonal we obtain the equation $e_1 =a_i Q\T$. It is clear that the first row of~$Q$ is given by~$a_i$. 
Thus, it is sufficient to compute an orthonormal basis including~$a_i$. This is possible in strongly polynomial time $O(d^3)=O(n^3)$ using the Gram-Schmidt process.

Since all Steps are repeated in this order at most $n$ times we obtain a running time $O(\frac{m^2 n^4}{\delta^2})$ for the repeated shadow vertex algorithm.
\begin{theorem}
The repeated shadow vertex algorithm has a running time of~$O(\frac{m^2 n^4}{\delta^2})$.
\end{theorem}

The entries of both~$c$ and~$\lambda$ in Algorithm~\ref{algorithm:SV} are continuous random variables. 
In practice it is, however, more realistic to assume that we can
draw a finite number of random bits. In Appendix~\ref{sec:RandomBits} we will show that our algorithm only needs to draw $\text{poly}(\log m, n,
\log(1/\delta))$ random bits in order to obtain the expected running time stated in Theorem~\ref{thm:MainNumberOfPivots} if~$\delta$ (or a good lower bound for it) is known.
However, if the parameter~$\delta$ is not known upfront and only discrete random variables
with a finite precision can be drawn, we have to modify the shadow vertex algorithm. This will give us an additional factor of $O(n)$ in the expected running time.

\section{Analysis of the Shadow Vertex Algorithm}\label{sec:analysis}

For given linear functions $L_1 \colon \RR^n \to \RR$ and $L_2 \colon \RR^n \to \RR$ we denote by $\pi = \pi_{L_1, L_2}$ the function $\pi \colon \RR^n \to \RR^2$, given by $\pi(x) = (L_1(x), L_2(x))$. Note that $n$-dimensional vectors can be treated as linear functions. By $P' = P'_{L_1, L_2}$ we denote the projection $\pi(P)$ of the polytope~$P$ onto the Euclidean plane, and by $R = R_{L_1, L_2}$ we denote the path from the bottommost vertex of~$P'$ to the rightmost vertex of~$P'$ along the edges of the lower envelope of~$P'$.

Our goal is to bound the expected number of edges of the path $R = R_{c, w}$, which is random since~$c$ and~$w$ are random. Each edge of~$R$ corresponds to a slope in $(0, \infty)$. These slopes are pairwise distinct with probability one (see Lemma~\ref{lemma:failure probability II}). Hence, the number of edges of~$R$ equals the number of distinct slopes of~$R$.

\begin{definition}
\label{definition:failure event}
For a real $\eps > 0$ let~$\F_\eps$ denote the event that there are three pairwise distinct vertices $z_1, z_2, z_3$ of~$P$ such that~$z_1$ and~$z_3$ are neighbors of~$z_2$ and such that
\[
  \left| \frac{w\T \cdot (z_2-z_1)}{c\T \cdot (z_2-z_1)} - \frac{w\T \cdot (z_3-z_2)}{c\T \cdot (z_3-z_2)} \right| \leq \eps \DOT
\]
\end{definition}

Note that if event~$\F_\eps$ does not occur, then all slopes of~$R$ differ by more than~$\eps$. Particularly, all slopes are pairwise distinct. 
First of all we show that event~$\F_\eps$ is very unlikely to occur if~$\eps$ is chosen sufficiently small. The proof of the following
lemma is almost identical to the corresponding proof in~\cite{BrunschR13} except that we need to adapt it to the different random model of~$c$. 
The proof as well as the proofs of some other lemmas that are almost identical to their counterparts in~\cite{BrunschR13}
can be found in Appendix~\ref{appendix:OmittedProofs} for the sake of completeness. Proofs that are completely identical to~\cite{BrunschR13}
are omitted.

\begin{lemma}
\label{lemma:failure probability II}
The probability of event~$\F_\eps$ tends to~$0$ for $\eps \to 0$.
\end{lemma}

Let~$p$ be a vertex of~$R$, but not the bottommost vertex $\pi(x_0)$. We call the slope~$s$ of the edge incident to~$p$ to the left of~$p$ \emph{the slope of~$p$}. As a convention, we set the slope of $\pi(x_0)$ to~$0$ which is smaller than the slope of any other vertex~$p$ of~$R$.

\begin{figure}
  \begin{center}
    \includegraphics[width=0.5\textwidth]{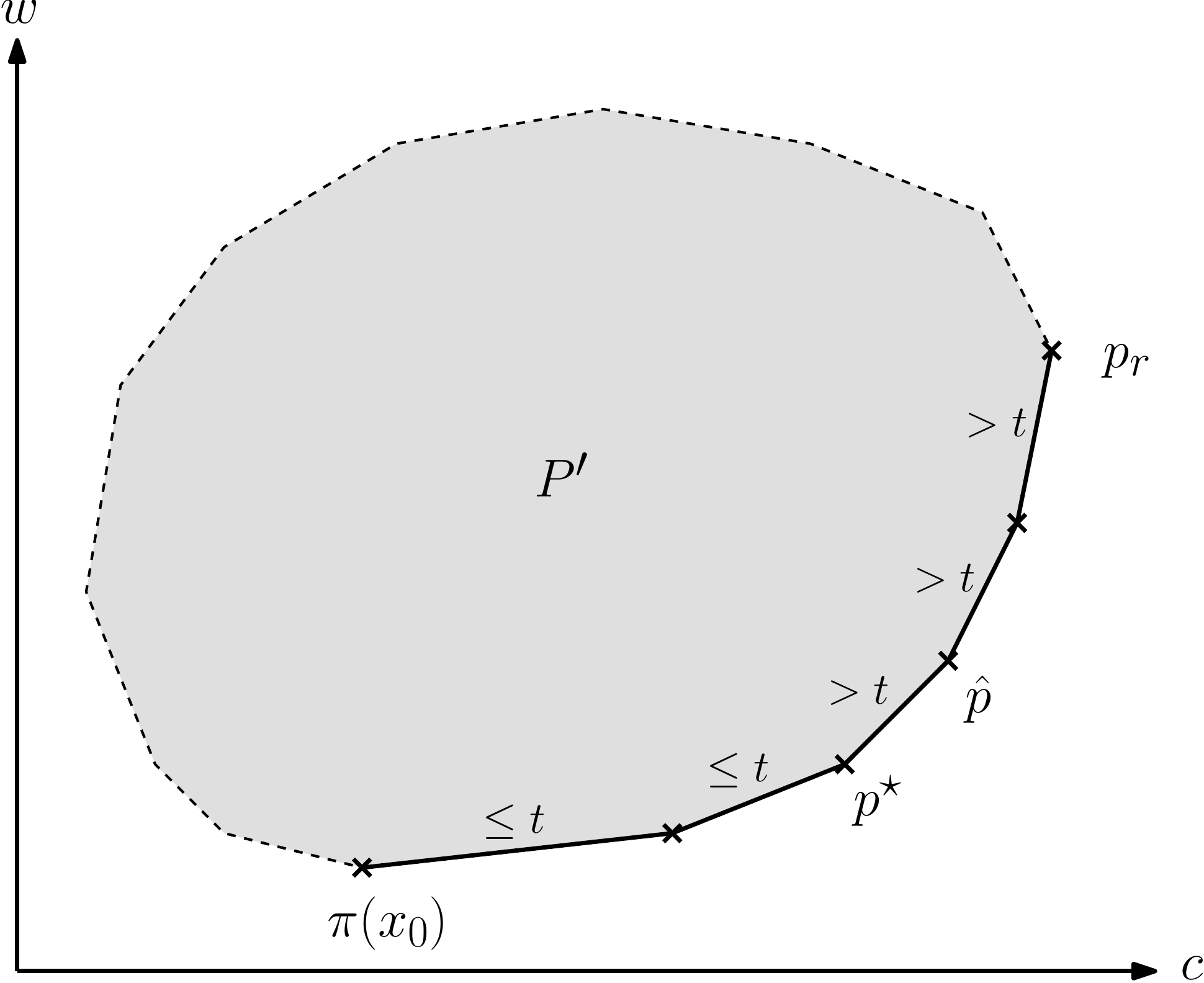}
  \end{center}
  \caption{Slopes of the vertices of~$R$}
  \label{fig:slopes}
\end{figure}

Let $t \geq 0$ be an arbitrary real, let~$p^\star$ be the rightmost vertex of~$R$ whose slope is at most~$t$, and let~$\hat{p}$ be the right neighbor of~$p^\star$, i.e., $\hat{p}$ is the leftmost vertex of~$R$ whose slope exceeds~$t$  (see Figure~\ref{fig:slopes}). Let~$x^\star$ and~$\hat{x}$ be the neighboring vertices of~$P$ with $\pi(x^\star) = p^\star$ and $\pi(\hat{x}) = \hat{p}$. Now let $i = i(x^\star, \hat{x}) \in [m]$ be the index for which $a_i\T x^\star = b_i$ and for which~$\hat{x}$ is the (unique) neighbor~$x$ of~$x^\star$ for which $a_i\T x < b_i$. This index is unique due to the non-degeneracy of the polytope~$P$. For an arbitrary real $\gamma \geq 0$ we consider the vector $\tilde{w} \DEF w - \gamma \cdot a_i$.

\begin{lemma}[Lemma~9 of~\cite{BrunschR13}]
\label{lemma:reconstruct}
Let $\tilde{\pi} = \pi_{c, \tilde{w}}$ and let $\tilde{R} = R_{c, \tilde{w}}$ be the path from $\tilde{\pi}(x_0)$ to the rightmost vertex~$\tilde{p}_r$ of the projection $\tilde{\pi}(P)$ of polytope~$P$. Furthermore, let~$\tilde{p}^\star$ be the rightmost vertex of~$\tilde{R}$ whose slope does not exceed~$t$. Then
$
  \tilde{p}^\star = \tilde{\pi}(x^\star)
$.
\end{lemma}

Let us reformulate the statement of Lemma~\ref{lemma:reconstruct} as follows: The vertex~$\tilde{p}^\star$ is defined for the path~$\tilde{R}$ of polygon $\tilde{\pi}(R)$ with the same rules as used to define the vertex~$p^\star$ of the original path~$R$ of polygon $\pi(P)$. Even though~$R$ and~$\tilde{R}$ can be very different in shape, both vertices, $p^\star$ and~$\tilde{p}^\star$, correspond to the same solution~$x^\star$ in the polytope~$P$, that is, $p^\star = \pi(x^\star)$ and $\tilde{p}^\star = \tilde{\pi}(x^\star)$.

Lemma~\ref{lemma:reconstruct} holds for any vector~$\tilde{w}$ on the ray $\vec{r} = \SET{ w - \gamma \cdot a_i \WHERE \gamma \geq 0}$. As $\|w\| \leq n$ (see Section~\ref{sec:Algorithm}), we have $w \in [-n,n]^n$. Hence, ray~$\vec{r}$ intersects the boundary of $[-n,n]^n$ in a unique point~$z$. We choose $\tilde{w} = \tilde{w}(w, i) \DEF z$ and obtain the following result.

\begin{corollary}
\label{corollary:reconstruct}
Let $\tilde{\pi} = \pi_{c, \tilde{w}(w, i)}$ and let~$\tilde{p}^\star$ be the rightmost vertex of path $\tilde{R} = R_{c, \tilde{w}(w, i)}$ whose slope does not exceed~$t$. Then
$
  \tilde{p}^\star = \tilde{\pi}(x^\star)
$.
\end{corollary}

Note that Corollary~\ref{corollary:reconstruct} only holds for the right choice of index $i = i(x^\star, \hat{x})$. However, the vector $\tilde{w}(w, i)$ can be defined for any vector $w \in [-n, n]^n$ and any index $i \in [m]$. In the remainder, index~$i$ is an arbitrary index from~$[m]$.

We can now define the following event that is parameterized in~$i$, $t$, and a real $\eps > 0$ and that depends on~$c$ and~$w$.

\begin{definition}
\label{definition:event E}
For an index $i \in [m]$ and a real $t \geq 0$ let~$\tilde{p}^\star$ be the rightmost vertex of $\tilde{R} = R_{c, \tilde{w}(w, i)}$ whose slope does not exceed~$t$ and let~$y^\star$ be the corresponding vertex of~$P$. For a real $\eps > 0$ we denote by~$\E_{i, t, \eps}$ the event that the conditions
\begin{itemize}[leftmargin=0.5cm]

  \item[$\bullet$] $a_i\T y^\star = b_i$ and

  \item[$\bullet$] $\frac{w\T (\hat{y} - y^\star)}{c\T (\hat{y} - y^\star)} \in (t, t+\eps]$, where~$\hat{y}$ is the neighbor~$y$ of~$y^\star$ for which $a_i\T y < b_i$,

\end{itemize}
are met. Note that the vertex~$\hat{y}$ always exists and that it is unique since the polytope~$P$ is non-degenerate.
\end{definition}

Let us remark that the vertices~$y^\star$ and~$\hat{y}$, which depend on the index~$i$, equal~$x^\star$ and~$\hat{x}$ if we choose $i = i(x^\star, \hat{x})$. For other choices of~$i$, this is, in general, not the case.

Observe that all possible realizations of~$w$ from the line $L \DEF \SET{ w + x \cdot a_i \WHERE x \in \RR }$ are mapped to the same vector $\tilde{w}(w, i)$. Consequently, if~$c$ is fixed and if we only consider realizations of~$\lambda$ for which $w \in L$, then vertex~$\tilde{p}^\star$ and, hence, vertex~$y^\star$ from Definition~\ref{definition:event E} are already determined. However, since~$w$ is not completely specified, we have some randomness left for event $\E_{i, t, \eps}$ to occur. This allows us to bound the probability of event $\E_{i, t, \eps}$ from above (see proof of Lemma~\ref{lemma:probability bound}). The next lemma shows why this probability matters.

\begin{lemma}[Lemma~12 from~\cite{BrunschR13}]
\label{lemma:event covering}
For any $t \geq 0$ and $\eps > 0$ let $\A_{t, \eps}$ denote the event that the path $R = R_{c, w}$ has a slope in $(t, t+\eps]$. Then, $\A_{t, \eps} \subseteq \bigcup_{i=1}^m \E_{i, t, \eps}$.
\end{lemma}

With Lemma~\ref{lemma:event covering} we can now bound the probability of event $\A_{t, \eps}$.
The proof of the next lemma is almost identical to the proof of Lemma~13 from~\cite{BrunschR13}.
We include it in the appendix for the sake of completeness. The only differences to Lemma~13 from~\cite{BrunschR13}
are that we can now use the stronger upper bound $\|c\|\le 2$ instead of $\|c\|\le n$ and that
we have more carefully analyzed the case of large~$t$. 

\begin{lemma}
\label{lemma:probability bound}
For any~$\phi\ge\sqrt{n}$, any~$t \geq 0$, and any $\eps > 0$ the probability of event $\A_{t, \eps}$ is bounded by
\[
  \Pr{\A_{t, \eps}} \leq \frac{2mn^2\eps}{\max \SET{ \frac n2, t } \cdot \delta^2} 
  \leq \frac{4mn\eps}{\delta^2}\DOT
\]
\end{lemma}

\begin{lemma}
\label{lemma:expectation bound}
For any interval~$I$ let~$X_{I}$ denote the number of slopes of $R = R_{c,w}$ that lie in the interval~$I$.
Then, for any~$\phi\ge\sqrt{n}$,
\[
  \Ex{X_{(0,n]}} \leq \frac{4mn^2}{\delta^2}
\]
\end{lemma}

\begin{proof}
For a real $\eps > 0$ let~$\F_\eps$ denote the event from Definition~\ref{definition:failure event}. 
Recall that all slopes of~$R$ differ by more than~$\eps$ if~$\F_\eps$ does not occur.
For~$t\in\mathbb{R}$ and~$\eps>0$ let~$Z_{t,\eps}$ be the random variable that indicates whether~$R$ has a slope in the 
interval $(t,t+\eps]$ or not, i.e., $Z_{t,\eps} = 1$ if~$X_{(t,t+\eps]}>0$ and $Z_{t,\eps} = 0$ if~$X_{(t,t+\eps]}=0$.

Let~$k \geq 1$ be an arbitrary integer.
We subdivide the interval~$(0,n]$ into~$k$ subintervals. If none of them contains more than one slope
then the number~$X_{(0,n]}$ of slopes in the interval~$(0,n]$ equals the number of subintervals for 
which the corresponding $Z$-variable equals~1. Formally
\[
  X_{(0,n]} \leq \begin{cases}
    \sum_{i=0}^{k-1} Z_{i\cdot\frac{n}{k},\frac{n}{k}} & \text{if $\F_{\frac{n}{k}}$ does not occur} \COMMA \cr
    m^n & \text{otherwise} \DOT
  \end{cases}
\]
This is true because $\binom{m}{n-1} \leq m^n$ is a worst-case bound on the number of edges of~$P$ and, hence, of the number of slopes of~$R$. 
Consequently,
\begin{align*}
  \Ex{X_{(0,n]}}
  &\leq \sum_{i=0}^{k-1} \Ex{Z_{i \cdot \frac{n}{k}, \frac{n}{k}}} + \Pr{\F_{\frac{n}{k}}} \cdot m^n
  = \sum_{i=0}^{k-1} \Pr{\A_{i \cdot \frac{n}{k}, \frac{n}{k}}} + \Pr{\F_{\frac{n}{k}}} \cdot m^n \cr
  &\leq \sum_{i=0}^{k-1} \frac{2mn^2 \cdot \frac{n}{k}}{\frac n2\delta^2} + \Pr{\F_{\frac{n}{k}}} \cdot m^n
  = \frac{4mn^{2}}{\delta^2} + \Pr{\F_{\frac{n}{k}}} \cdot m^n \DOT
\end{align*}
The second inequality stems from Lemma~\ref{lemma:probability bound}. 
Now the lemma follows because the bound on $\Ex{X_{(0,n]}}$ holds for any integer $k \geq 1$ and since $\Pr{\F_\eps} \to 0$ for $\eps \to 0$ in accordance with Lemma~\ref{lemma:failure probability II}.
\end{proof}

In~\cite{BrunschR13} we only computed an upper bound for the expected value of~$X_{(0,1]}$. Then we argued that the same
upper bound also holds for the expected value of~$X_{(1,\infty)}$. In order to see this, we simply exchanged the order  
of the objective functions in the projection~$\pi$. Then any edge with a slope of~$s>1$ becomes an edge with 
slope~$\frac{1}{s}<1$. Hence the number of slopes in~$[1,\infty)$ equals the number of slopes in~$(0,1]$ in the scenario
in which the objective functions are exchanged. Due to the symmetry in the choice of the objective functions in~\cite{BrunschR13}
the same analysis as before applies also to that scenario. 

We will now also exchange the order of the objective functions~$w\T x$ and~$c\T x$ in the projection. Since these objective functions
are not anymore
generated by the same random experiment, a simple argument as in~\cite{BrunschR13} is not possible anymore. Instead we
have to go through the whole analysis again. We will use the superscript~$-1$ to indicate that we are referring to the scenario
in which the order of the objective functions is exchanged. In particular, we consider the events~$\F_{\eps}^{-1}$, $\A_{t,\eps}^{-1}$,
and~$\E_{i,t,\eps}^{-1}$ that are defined analogously to their counterparts without superscript except that the order of the 
objective functions is exchanged.
The proof of the following lemma is analogous to the proof of Lemma~\ref{lemma:failure probability II}.
\begin{lemma}
\label{lemma:failure probability II2}
The probability of event~$\F^{-1}_\eps$ tends to~$0$ for $\eps \to 0$.
\end{lemma}

\begin{lemma}
\label{lemma:probability bound2}
For any~$\phi\ge\sqrt{n}$, any~$t \geq 0$, and any $\eps > 0$ the probability of event $\A^{-1}_{t,\eps}$ is bounded by
\[
  \Pr{\A_{t,\eps}^{-1}} \leq \frac{2mn^{3/2}\eps\phi}{\max \SET{ 1,\frac {nt}2} \cdot \delta}
  \le \frac{2mn^{3/2}\eps\phi}{\delta} \DOT
\]
\end{lemma}
\begin{proof}
Due to Lemma~\ref{lemma:event covering} (to be precise, due to its canonical adaption
to the events with superscript~$-1$) it suffices to show that
\[
  \Pr{\E_{i, t, \eps}^{-1}}
  \leq \frac{1}{m} \cdot \frac{2mn^{3/2}\eps\phi}{\max \SET{ 1,\frac {nt}2} \cdot \delta}
  = \frac{2n^{3/2}\eps\phi}{\max \SET{ 1,\frac {nt}2} \cdot \delta}
\]
for any index $i \in [m]$.

We apply the principle of deferred decisions and assume that vector~$w$ is already fixed.
Now we extend the normalized vector~$a_i$ to an orthonormal basis $\SET{ q_1, \ldots, q_{n-1}, a_i }$ of~$\RR^n$
and consider the random vector $(Y_1, \ldots, Y_{n-1}, Z)\T = Q\T c$ given by the matrix vector product
of the transpose of the orthogonal matrix $Q = [q_1, \ldots, q_{n-1}, a_i]$ and the vector $c = (c_1,\ldots,c_n)\T$.
For fixed values $y_1, \ldots, y_{n-1}$ let us consider all realizations of~$c$ such that $(Y_1, \ldots, Y_{n-1}) = (y_1, \ldots, y_{n-1})$. Then,~$c$ is fixed up to the ray
\[
  c(Z)
  = Q \cdot (y_1, \ldots, y_{n-1}, Z)\T
  = \sum_{j=1}^{n-1} y_j \cdot q_j + Z \cdot a_i
  = v + Z \cdot a_i
\]
for $v = \sum_{j=1}^{n-1} y_j \cdot q_j$. All realizations of $c(Z)$ that are under consideration are mapped to the same value~$\tilde{c}$ by the 
function $c \mapsto \tilde{c}(c, i)$, i.e., $\tilde{c}(c(Z), i) = \tilde{c}$ for any
possible realization of~$Z$. In other words, if $c = c(Z)$ is specified up to this ray, then the 
path $R_{\tilde{c}(c, i),w}$ and, hence, the vectors~$y^\star$ and~$\hat{y}$ from the definition of event 
$\E_{i, t, \eps}^{-1}$, are already determined.

Let us only consider the case that the first condition of event~$\E_{i, t, \eps}^{-1}$ is fulfilled. Otherwise, event~$\E_{i, t, \eps}$ cannot occur. Thus, event $\E_{i, t, \eps}^{-1}$ occurs iff
\[
  (t, t+\eps]
  \ni \frac{c\T \cdot (\hat{y} - y^\star)}{w\T \cdot (\hat{y} - y^\star)}
  = \underbrace{\frac{v\T \cdot (\hat{y} - y^\star)}{w\T \cdot (\hat{y} - y^\star)}}_{\FED \alpha} + Z \cdot \underbrace{\frac{a_i\T \cdot (\hat{y} - y^\star)}{w\T \cdot (\hat{y} - y^\star)}}_{\FED \beta} \DOT
\]
The next step in this proof will be to show that the inequality $|\beta| \geq \max \SET{ 1, \sqrt{n}\cdot t } \cdot \frac{\delta}{n}$ is necessary for event~$\E_{i, t, \eps}^{-1}$ to happen. For the sake of simplicity let us assume that $\|\hat{y} - y^\star\| = 1$ since~$\beta$ is invariant under scaling. If event~$\E_{i, t, \eps}^{-1}$ occurs, then $a_i\T y^\star = b_i$, $\hat{y}$ is a neighbor of~$y^\star$, and $a_i\T \hat{y} \neq b_i$. That is, by Lemma~\ref{lemma:delta properties}, Claim~\ref{delta properties:neighboring vertices} we obtain $|a_i\T \cdot (\hat{y} - y^\star)| \geq \delta \cdot \|\hat{y} - y^\star\| = \delta$ and, hence,
\[
  |\beta|
  = \left| \frac{a_i\T \cdot (\hat{y} - y^\star)}{w\T \cdot (\hat{y} - y^\star)} \right|
  \geq \frac{\delta}{|w\T \cdot (\hat{y} - y^\star)|} \DOT
\]
On the one hand we have $|w\T \cdot (\hat{y} - y^\star)| \leq \|w\| \cdot \|\hat{y}-y^\star\| \leq \Big(\sum_{i=1}^n\|u_i\|\Big) \cdot 1 \le n$.
On the other hand, due to $\frac{c\T \cdot (\hat{y} - y^\star)}{w\T \cdot (\hat{y} - y^\star)} \geq t$ we have
\[
  |w\T \cdot (\hat{y} - y^\star)|
  \leq \frac{|c\T \cdot (\hat{y} - y^\star)|}{t}
  \leq \frac{\|c\| \cdot \|\hat{y} - y^\star\|}{t}
  \leq \frac{\Big(1+\frac{\sqrt{n}}{\phi}\Big)}{t}
  \leq \frac{2}{t} \COMMA
\]
where the third inequality is due to the choice of~$c$ as perturbation of the
unit vector~$c_0$ and the fourth inequality is due to the assumption~$\phi\ge\sqrt{n}$.
Consequently,
\[
  |\beta|
  \geq \frac{\delta}{\min \SET{ n, \frac{2}{t} }}
  = \max \SET{ 1,\frac{nt}2 } \cdot \frac{\delta}{n} \DOT
\]
Summarizing the previous observations we can state that if event~$\E_{i, t, \eps}^{-1}$ occurs, then $|\beta| \geq \max \SET{ 1,\frac {nt}2} \cdot \frac{\delta}{n}$ and $\alpha + Z \cdot \beta \in (t, t+\eps]$. Hence,
\[
  Z \cdot \beta
  \in (t, t+\eps] - \alpha \COMMA
\]
i.e., $Z$ falls into an interval $I(y_1, \ldots, y_{n-1})$ of length at most $\eps/(\max \SET{ 1,\frac {nt}2} \cdot \delta/n) = n \eps/(\max \SET{ 1,\frac {nt}2} \cdot \delta)$ that only depends on the realizations $y_1, \ldots, y_{n-1}$ of $Y_1, \ldots, Y_{n-1}$. Let~$\B_{i, t, \eps}^{-1}$ denote the event that~$Z$ falls into the interval $I(Y_1, \ldots, Y_{n-1})$. We showed that $\E_{i, t, \eps}^{-1} \subseteq \B_{i, t, \eps}^{-1}$. Consequently,
\[
  \Pr{\E_{i, t, \eps}^{-1}}
  \leq \Pr{\B_{i, t, \eps}^{-1}}
  \leq \frac{2\sqrt{n} n \eps\phi}{\max \SET{ 1,\frac {nt}2}}
  \leq \frac{2n^{3/2}\eps\phi}{\max \SET{ 1,\frac {nt}2} \cdot \delta} \COMMA
\]
where the second inequality is due to Theorem~\ref{theorem.Prob:enough randomness} for
the orthogonal matrix~$Q$.
\end{proof}

\begin{lemma}
\label{lemma:expectation bound2}
For any interval~$I$ let~$X_{I}^{-1}$ denote the number of slopes of~$R_{w,c}$ that lie in the interval~$I$.
Then
\[
  \Ex{X^{-1}_{(0,1/n]}} \leq \frac{2m\sqrt{n}\phi}{\delta} \DOT  
\]
\end{lemma}
\begin{proof}
As in the proof of Lemma~\ref{lemma:expectation bound} we define 
for~$t\in\mathbb{R}$ and~$\eps>0$ the random variable~$Z_{t,\eps}^{-1}$ that indicates whether~$R_{w,c}$ has a slope in the 
interval $(t, t+\eps]$ or not. For any integer~$k\ge 1$ we obtain
\begin{align*}
  \Ex{X^{-1}_{\big(0,\frac{1}{n}\big]}}
  &\leq \sum_{i=0}^{k-1} \Ex{Z^{-1}_{i\cdot\frac{1}{kn},\frac{1}{kn}}} + \Pr{\F^{-1}_{\frac{1}{kn}}} \cdot m^n \cr
  & = \sum_{i=0}^{k-1} \Pr{\A^{-1}_{i\cdot\frac{1}{kn},\frac{1}{kn}}} + \Pr{\F^{-1}_{\frac{1}{kn}}} \cdot m^n \cr
  &\leq \sum_{i=0}^{k-1} \frac{2mn^{3/2}\phi}{kn\delta} + \Pr{\F^{-1}_{\frac{1}{k2^{\ell}\sqrt{n}}}} \cdot m^n
  = \frac{2m\sqrt{n}\phi}{\delta} + \Pr{\F^{-1}_{\frac{1}{k2^{\ell}\sqrt{n}}}} \cdot m^n \DOT
\end{align*}
The second inequality stems from Lemma~\ref{lemma:probability bound2}. 
Now the lemma follows because the bound holds for any integer $k \geq 1$ 
and $\Pr{\F^{-1}_\eps} \to 0$ for $\eps \to 0$ in accordance with Lemma~\ref{lemma:failure probability II2}.
\end{proof}

The following corollary directly implies Theorem~\ref{main}.

\begin{corollary}
\label{corollary:expectation bound}
The expected number of slopes of $R = R_{c,w}$ is 
\[
   \Ex{X_{(0,\infty)}} = \frac{4mn^2}{\delta^2}+\frac{2m\sqrt{n}\phi}{\delta}\DOT
\]
\end{corollary}
\begin{proof}
We divide the interval~$(0,\infty)$ into the subintervals~$(0,n]$
and~$(n,\infty)$. Using Lemma~\ref{lemma:expectation bound},
Lemma~\ref{lemma:expectation bound2}, and linearity of expectation we obtain
\begin{align*}
  \Ex{X_{(0,\infty)}} & = \Ex{X_{(0,n]}} + \Ex{X_{(n,\infty)}}
   = \Ex{X_{(0,n]}} + \Ex{X^{-1}_{\big(0,\frac{1}{n}\big]}} \cr
   & \le   \frac{4mn^2}{\delta^2} + \frac{2m\sqrt{n}\phi}{\delta} \DOT      
\end{align*}
In the second step we have exploited that by definition~$X_{(a,b)}=X^{-1}_{(1/b,1/a)}$ for any interval~$(a,b)$.
\end{proof}

\section{Finding a Basic Feasible Solution}\label{initial}

In this section we discuss how Phase~1 can be realized. In general there are, of course, several known textbook methods
how Phase~1 can be implemented. However, for our purposes it is crucial that the parameter~$\delta$ (or~$\Delta$) is not
too small (or too large) for the linear program that needs to be solved in Phase~1. Ideally we would like it to be identical with the parameter~$\delta$
(or~$\Delta$) of the matrix~$A$ of the original linear program. Eisenbrand and Vempala have addressed this problem and
have presented a method to implement Phase~1. Their method is, however, very different from usual textbook methods and
needs to solve~$m$ different linear programs to find an initial feasible solution for the given linear program.

In this section we will argue that also one of the usual textbook methods can 
be applied. We argue that $1/\delta$ increases by a factor of at most~$\sqrt{m}$ and that~$\Delta$ does not change at all
in case one considers integer matrices (in particular, for totally unimodular matrices).   

Let~$m$ and~$n$ be arbitrary positive integers, let $A \in \RR^{m \times n}$ be an arbitrary matrix without zero-rows, and let $c \in \RR^n$ and $b \in \RR^m$ be arbitrary vectors. For finding a basic feasible solution of the linear program
\begin{align*}
\text{(LP)} \left\{ \begin{aligned}
	\max &\ c\T x \cr
	\text{s.t.} &\ Ax \leq b
\end{aligned} \right.
\end{align*}
if one exists, or detecting that none exists, otherwise, we can solve the following linear program:
\begin{align*}
\text{(LP')} \left\{ \begin{aligned}
  \min &\sum_{i=1}^m y_i \cr
  \text{s.t.} &\ Ax - y \leq b \cr
              &\ y \geq 0
\end{aligned} \right.
\end{align*}
In the remainder of this section let us assume that matrix~$A$ has full column rank, that is, $\rank(A) = n$. Otherwise, we can transform the linear program (LP) as stated in Section~\ref{dimension} before considering (LP'). Furthermore, let us assume that the matrix~$\bar{A}$, formed by the first~$n$ rows of matrix~$A$, is invertible. This entails no loss of generality as this can always be achieved by permuting the rows of matrix~$A$.

Let~$\bar{b}$ denote the vector given by the first~$n$ entries of vector~$b$ and let~$\bar{x}$ denote the vector for which $\bar{A} \bar{x} = \bar{b}$. The vector $(x',y') = (\bar{x}, \max \{A\bar{x}-b, \NULL\})$ is a feasible solution of (LP'), where the maximum is meant component-wise and~$\NULL$ denotes the $m$-dimensional null vector. This is true because $Ax'-y' \leq A\bar{x} - (A\bar{x}-b) = b$ and $y' \geq \NULL$. Moreover, $(x',y')$ is a basic solution: By the choice of~$\bar{x}$ the first~$n$ inequalities of $Ax-y \leq b$ are tight as well as the first~$n$ non-negativity constraints. For each $k > m$ the $k\th$ inequality of $Ax-y \leq b$ or the $k\th$ non-negativity constraint is tight. Hence, the number of tight constraints is at least $2n+(m-n) = m+n$, which equals the number of variables of (LP').

Finally, we observe that a vector $(x, \NULL)$ is a basic feasible solution of (LP') if and only if~$x$ is a basic feasible solution of (LP). Consequently, by solving the linear program (LP') we obtain a basic feasible solution of the linear program (LP) (if the optimal value is~$0$) or we detect that (LP) is infeasible (if the optimal value is larger than~$0$). The linear program (LP') can be solved as described in Section~\ref{sec:Algorithm}. However, the running time is now expressed in the parameters $m'=2m$, $n'=m+n$ and $\delta(B)$ (or $\Delta(B)$) of the matrix
\[
  B = \begin{bmatrix}
    A                 & -\ID[m] \cr
    \ZERO[m \times n] & -\ID[m]
  \end{bmatrix}
  \in \RR^{(m+m) \times (n+m)} \DOT
\]
Before analyzing the parameters $\delta(B)$ and $\Delta(B)$, let us show that matrix~$B$ has full column rank.

\begin{lemma}
\label{lemma:Phase I full rank}
The rank of matrix~$B$ is $m+n$.
\end{lemma}

\begin{proof}
Recall that we assumed that the matrix~$\bar{A}$ given by the first~$n$ rows of matrix~$A$ is invertible. Now consider the first~$n$ rows and the last~$m$ rows of matrix~$B$. These rows form a submatrix~$\bar{B}$ of~$B$ of the form
\[
  \bar{B}
  = \begin{bmatrix}
    \bar{A} 					& C \cr
    \ZERO[m \times n] & -\ID[m]
  \end{bmatrix}
\]
for $C = [-\ID[n \times n], \ZERO[n \times (m-n)]]$. As~$\bar{B}$ is a $2 \times 2$-block-triangular matrix, we obtain $\det(\bar{B}) = \det(\bar{A}) \cdot \det(-\ID[n]) \neq 0$, that is, the first~$n$ rows and the last~$m$ rows of matrix~$B$ are linearly independent. Hence, $\rank(B) = m+n$.
\end{proof}

The remainder of this section is devoted to the analysis of $\delta(B)$ and $\Delta(B)$, respectively.

\subsection[A lower Bound for delta(B)]{A Lower Bound for $\delta(B)$}
\label{Phase I delta}

Before we derive a bound for the value $\delta(B)$, let us give a characterization of $\delta(M)$ for a matrix~$M$ with full column rank.

\newcommand{\ATOP}[2]{\genfrac{}{}{0pt}{}{#1}{#2}}
\begin{lemma}
\label{lemma:delta characterization}
Let $M \in \RR^{m \times n}$ be a matrix with rank~$n$. Then
\[
  \frac{1}{\delta(M)} = \max_{k \in [n]} \max \SET{ \|z\| \WHERE \ATOP{r_1\T, \ldots, r_n\T\ \text{linear independent rows}}{\text{of~$M$ and}\ [\N(r_1), \ldots, \N(r_n)]\T \cdot z = e_k } } \COMMA
\]
where~$e_k$ denotes the $k\th$ unit vector.
\end{lemma}

\begin{proof}
The correctness of the above statement follows from
\begin{align*}
  \frac{1}{\delta(M)}
  &= \max \SET{ \frac{1}{\delta(r_1, \ldots, r_n)} \WHERE r_1\T, \ldots, r_n\T\ \text{lin. indep. rows of~$M$} } \cr
  &= \max \SET{ \frac{1}{\delta(\N(r_1), \ldots, \N(r_n))} \WHERE r_1\T, \ldots, r_n\T\ \text{lin. indep. rows of~$M$} } \cr
  &= \max \SET{ \max_{k \in [n]} \|v_k\| \WHERE \ATOP{r_1\T, \ldots, r_n\T\ \text{lin. indep. rows of~$M$ and}}{[v_1, \ldots, v_n]^{-1} = [\N(r_1), \ldots, \N(r_n)]\T} } \DOT
\end{align*}
The first equation is due to the definition of~$\delta$, the second equation holds as~$\delta$ is invariant under scaling of rows, and the third equation is due to Claim~\ref{delta properties:inverse} of Lemma~\ref{lemma:delta properties}. The vector~$v_k$ from the last line is exactly the vector~$z$ for which $[\N(r_1), \ldots, \N(r_n)]\T \cdot z = e_k$. This finishes the proof.
\end{proof}

For the following lemma let us without loss of generality assume that the rows of matrix~$A$ are normalized. This does neither change the rank of~$A$ nor the value $\delta(A)$.

\begin{lemma}
\label{lemma:Phase I delta}
Let~$A$ and~$B$ be matrices of the form described above. Then
\[
  \frac{1}{\delta(B)} \leq \frac{2\sqrt{m-n+1}}{\delta(A)} \DOT
\]
\end{lemma}

\begin{proof}
In accordance with Lemma~\ref{lemma:delta characterization}, it suffices to show that for any~$m+n$ linearly independent rows $r_1\T, \ldots, r_{m+n}\T$ of~$B$ and any $k = 1, \ldots, m+n$ the inequality
\[
	\|z\| \leq \frac{2\sqrt{m-n+1}}{\delta(A)}
\]
holds, where~$z$ is the vector for which $[\N(r_1), \ldots, \N(r_{m+n})]\T \cdot z = e_k$.

Let $r_1\T, \ldots, r_{m+n}\T$ be arbitrary $m+n$ linearly independent rows of~$B$ and let $k \in [m+n]$ be an arbitrary integer. We consider the equation $\hat{B} \cdot z = e_k$, where $\hat{B} = [\N(r_1), \ldots, \N(r_{m+n})]\T$. Each row~$r_\ell$ is of either one of the two following types: Type~1 rows correspond to a row from~$A$ and for these we have $\|r_\ell\| = 2$ as the rows of~$A$ are normalized. Type~2 rows correspond to a non-negativity constraint of a variable~$y_i$. For these we have $\|r_\ell\| = 1$. Observe that each row has exactly one ``$-1$''-entry within the last~$m$ columns.

We categorize type~1 and type~2 rows further depending on the other selected rows: Type~1a rows are type~1 rows for which a type~2 row exists among the rows $r_1, \ldots, r_{m+n}$ which has its ``$-1$''-entry in the same column. This type~2 row is then classified as a type~2a row. The remaining type~1 and type~2 rows are classified as type~1b and type~2b rows, respectively. Observe that we can permute the rows of matrix~$\hat{B}$ arbitrarily as we show the claim for all unit vectors~$e_k$. Furthermore, we can permute the columns of~$\hat{B}$ arbitrarily because this only permutes the rows of the solution vector~$z$. This does not influence its norm. Hence, without loss of generality, matrix~$\hat{B}$ contains normalizations of type~1a, of type~2a, of type~1b, and of type~2b rows in this order and the normalizations of the type~2a rows are ordered the same way as the normalizations of their corresponding type~1a rows.

Let~$m_1$, $m_2$, and~$m_3$ denote the number of type~1a, type~1b, and type~2b rows, respectively. Observe that the number of type~2a rows is also~$m_1$. As matrix~$\hat{B}$ is invertible, each column contains at least one non-zero entry. Hence, we can permute the columns of~$\hat{B}$ such that~$\hat{B}$ is of the form
\[
  \hat{B}
  = \begin{bmatrix}
    \frac{1}{2} A_1 & -\frac{1}{2} \ID[m_1] & \ZERO                      & \ZERO \cr
    \ZERO           & -\ID[m_1]             & \ZERO                      & \ZERO \cr
    \frac{1}{2} A_2 & \ZERO                & -\frac{1}{2} \ID[m_2] & \ZERO \cr
    \ZERO           & \ZERO                & \ZERO                      & -\ID[m_3]
  \end{bmatrix}
  \in \RR^{(m+n) \times (m+n)} \COMMA
\]
where~$A_1$ and~$A_2$ are $m_1 \times n$- and $m_2 \times n$-submatrices of~$A$, respectively. The number of rows of~$\hat{B}$ is $2m_1+m_2+m_3 = m+n$, whereas the number of columns of~$\hat{B}$ is $n+m_1+m_2+m_3=m+n$. This implies $m_1 = n$ and $m_2 \leq m-n$. Particularly, $A_1$ is a square matrix. As matrix~$\hat{B}$ is a $2 \times 2$-block-triangular matrix and the top left and the bottom right block are $2 \times 2$-block-triangular matrices as well, we obtain
\[
	\det(\hat{B})
  = \det \left( \frac{1}{2} A_1 \right) \cdot (-1)^{m_1} \cdot \left( -\frac{1}{2} \right)^{m_2} \cdot (-1)^{m_3}
  = \pm\det(A_1) \cdot \frac{1}{2^{n+m_2}} \DOT
\]
Due to the linear independence of the rows $r_1\T, \ldots, r_{m+n}\T$ we have $\det(\hat{B}) \neq 0$. Consequently, $\det(A_1) \neq 0$, that is, matrix~$A_1$ is invertible.

We partition vector~$z$ and vector~$e_k$ into four components $z_1, \ldots, z_4$ and $e_k^{(1)}, \ldots, e_k^{(4)}$, respectively, and rewrite the system $\hat{B} \cdot z = e_k$ of linear equations as follows:
\begin{align*}
  \frac{1}{2} A_1 z_1 - \frac{1}{2} z_2 &= e_k^{(1)} \cr
  -z_2 &= e_k^{(2)} \cr
  \frac{1}{2} A_2 z_1 - \frac{1}{2} z_3 &= e_k^{(3)} \cr
  -z_4 &= e_k^{(4)}
\end{align*}
Now we distinguish between four pairwise distinct cases $e_k^{(i)} \neq \NULL$ for $i = 1, \ldots, 4$. In any case recall that the rows of~$A_1$ and~$A_2$ are rows of~$A$, which are normalized. Furthermore, recall that the rows of~$A_1$ are linearly independent.
\begin{itemize}[leftmargin=0.5cm]

	\item \textbf{Case~1:} $e_k^{(1)} \neq \NULL$. In this case we obtain $z_2 = \NULL$ and $z_4 = \NULL$. This implies $z_1 = 2\hat{z}$, where~$\hat{z}$ is the solution of the equation $A_1 \hat{z} = e_k^{(1)} + \frac{1}{2} \cdot \NULL = e_k^{(1)}$. As the rows of matrix~$A_1$ are normalized, Lemma~\ref{lemma:delta characterization} yields $\|\hat{z}\| \leq 1/\delta(A)$ and, hence, $\|z_1\| \leq 2/\delta(A)$. Next, we obtain $z_3 = A_2 z_1 - 2 \cdot e_k^{(3)} = A_2 z_1 - \NULL = A_2 z_1$. Each entry of~$z_3$ is a dot product of a (normalized) row from~$A$ and~$z_1$. Hence, the absolute value of each entry is bounded by $\|z_1\| \leq 2/\delta(A)$. This yields the inequality
	\begin{align*}
	  \|z\|
	  &= \sqrt{\|z_1\|^2 + \|z_2\|^2 + \|z_3\|^2 + \|z_4\|^2}
	  \leq \sqrt{(1+m_2) \cdot (2/\delta(A))^2} \cr
	  &\leq \frac{2\sqrt{m-n+1}}{\delta(A)} \DOT
	\end{align*}
	For the last inequality we used the fact that $m_2 \leq m-n$.

	\item \textbf{Case~2:} $e_k^{(2)} \neq \NULL$. Here we obtain $z_2 = -e_k^{(2)}$, $z_4 = \NULL$, and $A_1 z_1 = 2 \cdot e_k^{(1)} + z_2 = 2 \cdot \NULL - e_k^{(2)} = -e_k^{(2)}$, that is, $z_1 = -\hat{z}$, where~$\hat{z}$ is the solution of the equation $A_1 \hat{z} = e_k^{(2)}$. Analogously as in Case~1, we obtain $\|\hat{z}\| \leq 1/\delta(A)$ and, hence, $\|z_1\| \leq 1/\delta(A)$. Moreover, we obtain $z_3 = A_2 z_1 - 2 \cdot e_k^{(3)} = A_2 z_1 - \NULL = A_2 z_1$, that is, the absolute value of each entry of~$z_3$ is bounded by $\|z_1\| \leq 1/\delta(A)$. Consequently,
	\begin{align*}
	  \|z\|
	  &\leq \sqrt{1+(1+m_2) \cdot (1/\delta(A))^2}
	  \leq \frac{\sqrt{m-n+2}}{\delta(A)}
	  \leq \frac{2\sqrt{m-n+1}}{\delta(A)} \DOT
	\end{align*}
	For the second inequality we used $m_2 \leq m-n$ and $\delta(A) \leq 1$ by definition of~$\delta(A)$. In the last inequality we used the fact that $m-n+1 \geq 1$ and $\sqrt{x+1} \leq 2\sqrt{x}$ for all $x \geq 1/3$.

	\item \textbf{Case~3:} $e_k^{(3)} \neq \NULL$. In this case we obtain $z_2 = \NULL$, $z_4 = \NULL$, and hence, $z_1 = \NULL$. This yields $z_3 = -2 \cdot e_k^{(3)}$ and
	\[
	  \|z\|
	  = \|z_3\|
	  = 2
	  \leq \frac{2\sqrt{m-n+1}}{\delta(A)} \COMMA
	\]
	where we again used $\delta(A) \leq 1$.

	\item \textbf{Case~4:} $e_k^{(4)} \neq \NULL$. Here we obtain $z_2 = \NULL$, $z_4 = -e_k^{(4)}$, and hence, $z_1 = \NULL$ and $z_3 = \NULL$. Consequently, we get
	\[
	  \|z\|
	  = \|z_4\|
	  = 1
	  \leq \frac{2\sqrt{m-n+1}}{\delta(A)} \COMMA
	\]
	which completes this case distinction.
\end{itemize}
As we have seen, in any case the inequality $\|z\| \leq 2\sqrt{m-n+1}/\delta(A)$ holds, which finishes the proof.
\end{proof}

\subsection[An Upper Bound for Delta(B)]{An Upper Bound for $\Delta(B)$}
\label{Phase I Delta}

Although parameter $\Delta(B)$ can be defined for arbitrary real-valued matrices, its meaning is limited to integer matrices when considering our analysis of the expected running time of the shadow vertex method. Hence, in this section we only deal with the case that matrix~$A$ is integral. Unlike in Section~\ref{Phase I delta}, we do not normalize the rows of matrix~$A$ before considering the linear program (LP'). As a consequence, matrix~$B$ is also integral.

The following lemma establishes a connection between $\Delta(A)$ and $\Delta(B)$.

\begin{lemma}
\label{lemma:Phase I Delta}
Let~$A$ and~$B$ be of the form described above. Then $\Delta(B) = \Delta(A)$.
\end{lemma}

\begin{proof}
It is clear that $\Delta(B) \geq \Delta(A)$ as matrix~$B$ contains matrix~$A$ as a submatrix. Thus, we can concentrate on proving that $\Delta(B) \leq \Delta(A)$. For this, consider an arbitrary $k \times k$-submatrix~$\hat{B}$ of~$B$. Matrix~$\hat{B}$ is of the form
\[
  \hat{B}
  = \begin{bmatrix}
    A'                        & -I_1 \cr
    \ZERO[k_1 \times (k-k_2)] & -I_2 \cr
  \end{bmatrix} \COMMA
\]
where~$A'$ is a $(k-k_1) \times (k-k_2)$-submatrix of~$A$ and~$I_1$ and~$I_2$ are $(k-k_1) \times k_2$- and $k_1 \times k_2$-submatrices of $\ID[m]$, respectively. Our goal is to show that $|\det(\hat{B})| \leq \Delta(A)$. By analogy with the proof of Lemma~\ref{lemma:Phase I delta} we partition the rows of~$\hat{B}$ into classes. A row of~$\hat{B}$ is of type~1 if it contains a row from~$A'$. Otherwise, it is of type~2. Consequently, there are $k-k_1$ type~1 and~$k_1$ type~2 rows.

These type~1 and type~2 rows are further categorized into three subtypes depending on the ``$-1$''-entry (if exists) within the last~$k_2$ columns. Type~1 and type~2 rows that only have zeros in the last~$k_2$ entries are classified as type~1c and type~2c rows, respectively. The remaining type~1 and type~2 rows have exactly one ``$-1$''-entry within the last~$k_2$ columns. These are partitioned into subclasses as follows: If there are a type~1 row and a type~2 row that have their ``$-1$''-entry in the same column, then these rows are classified as type~1a and type~2a, respectively. The type~1 and type~2 rows that are neither type~1a nor type~1c nor type~2a nor type~2c are referred to as type~1b and type~2b rows, respectively.

Note that type~2c rows only contain zeros. If matrix~$\hat{B}$ contains such a row, then $|\det(\hat{B})| = 0 \leq \Delta(A)$. Hence, in the remainder we only consider the case that matrix~$\hat{B}$ does not contain type~2c rows. With the same argument we can assume, without loss of generality, that matrix~$\hat{B}$ does not contain a column with only zeros. As permuting the rows and columns of matrix~$\hat{B}$ does not change the absolute value of its determinant, we can assume that~$\hat{B}$ contains type~1a, type~1c, type~2a, type~1b, and type~2b rows in this order and that the type~2a rows are ordered the same ways as their corresponding type~1a rows. Furthermore, we can permute the columns of~$\hat{B}$ such that it has the following form:
\[
  \hat{B}
  = \begin{bmatrix}
    A_1   & -\ID  & \ZERO & \ZERO \cr
    A_2   & \ZERO & \ZERO & \ZERO \cr
    \ZERO & -\ID  & \ZERO & \ZERO \cr
    A_3   & \ZERO & -\ID  & \ZERO \cr
    \ZERO & \ZERO & \ZERO & -\ID
  \end{bmatrix} \COMMA
\]
where~$A_1$, $A_2$, and~$A_3$ are submatrices of~$A'$ and, hence, of~$A$. Iteratively decomposing matrix~$\hat{B}$ into blocks and exploiting the block-triangular form of the matrices obtained in each step yields
\begin{align*}
  |\det(\hat{B})|
  &= \left| \det \left( \begin{bmatrix}
    A_1   & -\ID \cr
    A_2   & \ZERO \cr
    \ZERO & -\ID
  \end{bmatrix} \right) \right|
  \cdot
  \left| \det \left( \begin{bmatrix}
    -\ID  & \ZERO \cr
    \ZERO & -\ID
  \end{bmatrix} \right) \right|
  = \left| \det \left( \begin{bmatrix}
    A_1   & -\ID \cr
    A_2   & \ZERO \cr
    \ZERO & -\ID
  \end{bmatrix} \right) \right| \cr
  &= \left| \det \left( \begin{bmatrix}
    A_1 \cr
    A_2
  \end{bmatrix} \right) \right|
  \cdot
  |\det(-\ID)|
  = \left| \det \left( \begin{bmatrix}
    A_1 \cr
    A_2
  \end{bmatrix} \right) \right| \DOT
\end{align*}
The absolute value of the latter determinant is bounded from above by~$\Delta(A)$. This completes the proof.
\end{proof}
\section{Conclusions}

We have shown that the shadow vertex algorithm can be used to solve linear programs possessing the $\delta$-distance property
in strongly polynomial time with respect to $n$, $m$, and~$1/\delta$. The bound we obtained in 
Theorem~\ref{thm:MainNumberOfPivots} depends quadratically on~$1/\delta$. Roughly speaking, one term~$1/\delta$ is due to  
the fact that the smaller~$\delta$ the less random is the objective function~$w\T x$. This term could in fact be replaced by~$1/\delta(B)$
where~$B$ is the matrix that contains only the rows that are tight for~$x$. The other term~$1/\delta$ is due to our application
of the principle of deferred decisions in the proof of Lemma~\ref{lemma:probability bound}. The smaller~$\delta$ the less random
is~$w(Z)$.

For packing linear programs, in which all coefficients of~$A$ and~$b$ are non-negative and one has~$x\ge 0$ as additional constraint,
it is, for example, clear that~$x=0^n$ is a basic feasible solution. That is, one does not need to run Phase~1. Furthermore	as in this
solution without loss of generality exactly the constraints~$x\ge 0$ are tight,~$\delta(B)=1$ and one occurrence of~$1/\delta$ in
Theorem~\ref{thm:MainNumberOfPivots} can be removed.

\section*{Acknowledgments}
The authors would like to thank Friedrich Eisenbrand and Santosh Vempala for providing detailed explanations of their paper
and the anonymous reviewers for valuable suggestions how to improve the presentation.

\bibliographystyle{plain}
\bibliography{literature}

\clearpage

\begin{appendix}

\section*{Appendix}

In Appendix~\ref{sdelta} we give an equivalent definition of~$\delta$ and state some important properties
that are used later. Appendix~\ref{sec:protheory} contains some theorems from probability theory that will be used
in Appendix~\ref{appendix:OmittedProofs}, which contains the omitted proofs from Section~\ref{sec:analysis}.
In Appendix~\ref{specialcases} we argue how to cope with unbounded linear programs and linear programs without full column rank.
We conclude with Appendix~\ref{sec:RandomBits} in which we analyze the number of random bits necessary to
run the shadow vertex method.

\section[The Parameter delta]{The Parameter $\delta$} \label{sdelta}

In~\cite{BrunschR13} we introduced the parameter $\delta$ only for $m \times n$-matrices~$A$ with rank~$n$. 
This was the only interesting case for the type of problem considered there. In this paper we cannot assume the constraint
matrix to have full column rank. Hence, in Definition~\ref{def:delta} we extended the definition of~$\delta$ to arbitrary matrices
(as Eisenbrand and Vempala~\cite{GRE}). We will now give a definition of~$\delta$ that is equivalent to Definition~\ref{def:delta}
and allows to prove some important properties of~$\delta$. 

\begin{definition}
\label{definition:delta}\leavevmode
\begin{enumerate}[leftmargin=0.8cm]
	\item Let $z_1, \ldots, z_k \in \RR^n$ be $k \geq 2$ linearly independent vectors and let $\varphi \in (0, \frac{\pi}{2}]$ be the angle between $z_k$ and $\SPAN{z_1, \ldots, z_{k-1}}$. By $\hat{\delta}(\SET{ z_1, \ldots, z_{k-1} }, z_k) = \sin \varphi$ we denote the sine of~$\varphi$. Moreover, we set
	\[
		\delta(z_1, \ldots, z_k) = \min_{\ell \in [k]} \hat{\delta}(\SET{ z_i \WHERE i \in [k] \setminus \SET{\ell} }, z_\ell) \DOT
	\]
	
	\item Given a matrix $A = [a_1, \ldots, a_m]\T \in \RR^{m \times n}$ with rank $r = \rank(A) \geq 2$, we set
	\[
		\delta(A) = \min \SET{ \delta(a_{i_1}, \ldots, a_{i_r}) \WHERE a_{i_1}, \ldots, a_{i_r}\ \text{linearly independent} } \DOT
	\]
\end{enumerate}
\end{definition}

Note that for the angle~$\varphi$ in Definition~\ref{definition:delta} we obtain the equation
\[
	\varphi = \min \SET{ \angle(z_k, z) \WHERE z \in \SPAN{z_1, \ldots, z_{k-1}} }\DOT
\]
Furthermore, the minimum is attained for the orthogonal projection of the vector~$z_k$ onto $\SPAN{z_1, \ldots, z_{k-1}}$ when we use the convention $\angle(x, \NULL) \DEF \frac{\pi}{2}$ for any vector $x \in \RR^n$. For this reason the sine is given by the length of the orthogonal projection  divided by $\|z_k\|$. In the case where $\|z_k\|$ has length $1$ this equals the length of the orthogonal projection and thus the $\delta$-distance of $z_k$ to $\SPAN{z_1, \ldots, z_{k-1}}$ as defined in Definition~\ref{def:delta}.

\begin{lemma}[Lemma~5 of~\cite{BrunschR13}]\label{deltaproperties}
\label{lemma:delta properties}
Let $z_1, \dots, z_n \in \RR^n$ be linearly independent vectors of length~$1$, let $A \in \RR^{m \times n}$ be a matrix with $\rank(A) = n$, and let~$\delta:=\delta(A)$. Then the following properties hold:
\begin{enumerate}[leftmargin=0.8cm]

  \item \label{delta properties:inverse} If~$M$ is the inverse of $[z_1, \ldots, z_n]\T$, then
  \vspace{-0.5em}\[
    \delta(z_1, \ldots, z_n)
    = \frac{1}{\max_{k \in [n]} \|m_k\|}
    \leq \frac{\sqrt{n}}{\max_{k \in [n]} \|M_k\|} \COMMA \vspace{-0.5em}
  \]
  where $[m_1, \ldots, m_n] = M$ and $[M_1, \ldots, M_n] = M\T$.
  
  \item \label{delta properties:orthogonal matrix} If $Q \in \RR^{n \times n}$ is an orthogonal matrix, then $\delta(Qz_1, \ldots, Qz_n) = \delta(z_1, \ldots, z_n)$.

  \item \label{delta properties:neighboring vertices} Let~$y_1$ and~$y_2$ be two neighboring vertices of $P = \SET{ x \in \RR^n \WHERE Ax \leq b }$ and let~$a_i\T$ be a row of~$A$. If $a_i\T \cdot (y_2-y_1) \neq 0$, then $|a_i\T \cdot (y_2-y_1)| \geq \delta \cdot \|y_2-y_1\|$.

  \item \label{delta properties:comparison} If~$A$ is an integral matrix, then $\frac{1}{\delta} \leq n \Delta_1 \Delta_{n-1} \leq n \Delta^2$, where~$\Delta$, $\Delta_1$, and $\Delta_{n-1}$ are the largest absolute values of any sub-determinant of~$A$ of arbitrary size, of size~$1$, and of size~$n-1$, respectively.
  
\end{enumerate}
\end{lemma}

\section{Some Probability Theory}
\label{sec:protheory}

In this section we state and formulate the corollary about linear combinations of random variables used in Section~\ref{sec:analysis}. This theorem follows from Theorem~3.3 of~\cite{Brunsch14} which we will recite here in a simplified variant.

\newcommand{\U}{\mathcal{U}}
\begin{theorem}[cf.\ Theorem~3.3 of~\cite{Brunsch14}]
\label{theorem.Prob:enough randomness}
Let $\eps > 0$ and $\phi \geq 1$ be reals, let $I_1, \ldots, I_n \subseteq [-1,1]$ be intervals of length~$1/\phi$, and let $X_1, \ldots, X_n$ be independent random variables such that~$X_k$ is uniformly distributed on~$I_k$ for $k = 1, \ldots, n$. Moreover, let $A \in \RR^{n \times n}$ be an invertible matrix, let $(Y_1, \ldots, Y_{n-1}, Z)\T = A \cdot (X_1, \ldots, X_n)\T$ be the linear combinations of $X_1, \ldots, X_n$ given by~$A$, and let $I \colon \RR^{n-1} \to \SET{ [x, x+\eps] \WHERE x \in \RR }$ be a function mapping a tuple $(y_1, \ldots, y_{n-1}) \in \RR^{n-1}$ to an interval $I(y_1, \ldots, y_{n-1})$ of length~$\eps$. Then the probability that~$Z$ falls into the interval $I(Y_1, \ldots, Y_{n-1})$ can be bounded by
\[
	\Pr{Z \in I(Y_1, \ldots, Y_{n-1})} \leq 2\eps\phi \cdot \sum_{i=1}^n \frac{|\det A_{n,i}|}{|\det A|} \COMMA
\]
where~$A_{n,i}$ is the $(n-1) \times (n-1)$-submatrix of~$A$ obtained from~$A$ by removing row~$n$ and column~$i$.
\end{theorem}

Now we can state 

\begin{corollary}
\label{corollary.Prob:enough randomness}
Let $\eps$, $\phi$, $X_1, \ldots, X_n$, $A$, $Y_1, \ldots, Y_{n-1}, Z$, and~$I$ be as in Theorem~\ref{theorem.Prob:enough randomness}. Then the probability that~$Z$ falls into the interval $I(Y_1, \ldots, Y_{n-1})$ can be bounded by
\[
  \Pr{Z \in I(Y_1, \ldots, Y_{n-1})} \leq \frac{2n\eps\phi}{\delta(a_1, \ldots, a_n) \cdot \min_{k \in [n]} \|a_k\|} \COMMA
\]
where $a_1, \ldots, a_n$ denote the columns of matrix~$A$. Furthermore, if~$A$ is orthogonal, then even the stronger bound
\[
  \Pr{Z \in I(Y_1, \ldots, Y_{n-1})} \leq 2\sqrt{n}\eps\phi
\]
holds.
\end{corollary}

\begin{proof}
In accordance with Theorem~\ref{theorem.Prob:enough randomness} it suffices to bound the sum
$
  \sum_{i=1}^n \frac{|\det(A_{n,i})|}{|\det(A)|}
$
from above. For this, consider the equation $Ax = e_n$, where $e_n = (0, \ldots, 0, 1) \in \RR^n$ denotes the $n\th$ unit vector. Following Cramer's rule and Laplace's formula, we obtain
\[
  |x_i|
  = \frac{|\det([a_1, \ldots, a_{i-1}, e_n, a_{i+1}, \ldots, a_n])|}{|\det(A)|}
  = \frac{|\det(A_{n,i})|}{|\det(A)|} \DOT
\]
Hence, applying Theorem~\ref{theorem.Prob:enough randomness} yields
\[
  \Pr{Z \in I(Y_1, \ldots, Y_{n-1})}
  \leq 2\eps\phi \cdot \sum_{i=1}^n |x_i|
  = 2\eps\phi \cdot \|x\|_1
  \leq 2\sqrt{n}\eps\phi \cdot \|x\| \DOT
\]
Recall, that by~$\|x\|$ we refer to the Euclidean norm~$\|x\|_2$ of~$x$. The claim for orthogonal matrices~$A$ follows immediately since $\|x\| = \|A^{-1} e_n\| = \|e_n\| = 1$ because $A^{-1} = A\T$ is orthogonal as well.

For the general case we consider the equation $\hat{A} \hat{x} = e_n$, where $\hat{A} = [\N(a_1), \ldots, \N(a_n)]$ consists of the normalized columns of matrix~$A$. Vector $\hat{x} = \hat{A}^{-1} e_n$ is the $n^\text{th}$ column of the matrix $\hat{A}^{-1}$. Thus, we obtain
\[
  \|\hat{x}\|
  \leq \max_{\substack{\text{$r$ column}\\\text{of $\hat{A}^{-1}$}}} \|r\|
  \leq \frac{\sqrt{n}}{\delta(a_1, \ldots, a_n)} \COMMA
\]
where second inequality is due to Claim~\ref{delta properties:inverse} of Lemma~\ref{lemma:delta properties}. Due to $A = \hat{A} \cdot \diag(\|a_1\|, \ldots, \|a_n\|)$, we have
\[
  x
  = A^{-1} e_n
  = \diag \left( \frac{1}{\|a_1\|}, \ldots, \frac{1}{\|a_n\|} \right) \cdot \hat{A}^{-1} e_n
  = \diag \left( \frac{1}{\|a_1\|}, \ldots, \frac{1}{\|a_n\|} \right) \cdot \hat{x} \DOT
\]
Consequently, $\|x\| \leq \|\hat{x}\|/\min_{k \in [n]} \|a_k\|$ and, thus,
\begin{align*}
  \Pr{Z \in I(Y_1, \ldots, Y_{n-1})}
  &\leq 2\sqrt{n}\eps\phi \cdot \frac{\|\hat{x}\|}{\min_{k \in [n]} \|a_k\|}
  \leq \frac{2n\eps\phi}{\delta(a_1, \ldots, a_n) \cdot \min_{k \in [n]} \|a_k\|} \DOT \qedhere
\end{align*}
\end{proof}

\section{Proofs from Section~\ref{sec:analysis}}\label{appendix:OmittedProofs}

In this section we give the omitted proofs from Section~\ref{sec:analysis}. These are merely contained for the sake of completeness
because they are very similar to the corresponding proofs in~\cite{BrunschR13}. 

\subsection{Proof of Lemma~\ref{lemma:failure probability II}}

\begin{lemma}
\label{lemma:failure probability I}
The probability that there are two neighboring vertices $z_1, z_2$ of~$P$ such that $|c\T \cdot (z_2-z_1)| \leq \eps \cdot \|z_2-z_1\|$ is bounded from above by $2m^n n\eps\phi$.
\end{lemma}

\begin{proof}
Let~$z_1$ and~$z_2$ be arbitrary points in $\RR^n$, let $u = z_2 - z_1$, and let~$A_\eps$ denote the event that $|c\T \cdot u| \leq \eps \cdot \|u\|$. As this inequality is invariant under scaling, we can assume that $\|u\| = 1$. Hence, there exists an index~$i$ for which $|u_i| \geq 1/\sqrt{n} \geq 1/n$. We apply the principle of deferred decisions and assume that the coefficients~$c_j$ for $j \neq i$ are already fixed arbitrarily. Then event~$A_\eps$ occurs if and only if $c_i \cdot u_i \in [-\eps, \eps] - \sum_{j \neq i} c_j u_j$. Hence, for event~$A_\eps$ to occur the random coefficient~$c_i$ must fall into an interval of length $2\eps/|u_i| \leq 2n\eps$. The probability for this is bounded from above by~$2n\eps\phi$.

As we have to consider at most $\binom{m}{n-1} \leq m^n$ pairs of neighbors $(z_1, z_2)$, a union bound yields the additional factor of $m^n$.
\end{proof}

\begin{proof}[Proof of Lemma~\ref{lemma:failure probability II}]
Let $z_1, z_2, z_3$ be pairwise distinct vertices of~$P$ such that~$z_1$ and~$z_3$ are neighbors of~$z_2$ and let $\Delta_z \DEF z_2-z_1$ and  $\Delta'_z \DEF z_3-z_2$. We assume that $\|\Delta_z\| = \|\Delta'_z\| = 1$. This entails no loss of generality as the fractions in Definition~\ref{definition:failure event} are invariant under scaling. Let $i_1, \ldots, i_{n-1} \in [m]$ be the $n-1$ indices for which $a_{i_k}\T z_1 = b_{i_k} = a_{i_k}\T z_2$. For the ease of notation let us assume that $i_k = k$. The rows $a_1, \ldots, a_{n-1}$ are linearly independent because~$P$ is non-degenerate. Since $z_1, z_2, z_3$ are distinct vertices of~$P$ and since~$z_1$ and~$z_3$ are neighbors of~$z_2$, there is exactly one index~$\ell$ for which $a_\ell\T z_3 < b_\ell$, i.e., $a_\ell\T \Delta'_z \neq 0$. Otherwise, $z_1, z_2, z_3$ would be collinear which would contradict the fact that they are pairwise distinct vertices of~$P$. Without loss of generality assume that $\ell = n-1$. Since $a_k\T \Delta_z = 0$ for each $k \in [n-1]$, the vectors $a_1, \ldots, a_{n-1}, \Delta_z$ are linearly independent.

We apply the principle of deferred decisions and assume that~$c$ is already fixed. Thus, $c\T \Delta_z$ and $c\T \Delta'_z$ are fixed as well. Moreover, we assume that $c\T \Delta_z \neq 0$ and $c\T \Delta'_z \neq 0$ since this happens almost surely due to Lemma~\ref{lemma:failure probability I}. Now consider the matrix $M = [a_1, \ldots, a_{n-2}, \Delta_z, a_{n-1}]$ and the random vector
$
  (Y_1, \ldots, Y_{n-1}, Z)\T
  = M^{-1} \cdot w
  = -M^{-1} \cdot [u_1, \ldots, u_n] \cdot \lambda
$.
For fixed values $y_1, \ldots, y_{n-1}$ let us consider all realizations of~$\lambda$ for which $(Y_1, \ldots, Y_{n-1}) = (y_1, \ldots, y_{n-1})$. Then
\begin{align*}
  w\T \Delta_z
  &= \big( M \cdot (y_1, \ldots, y_{n-1}, Z)\T \big)\T \Delta_z \cr
  &= \sum_{k=1}^{n-2} y_k \cdot a_k\T \Delta_z + y_{n-1} \cdot \Delta_z\T \Delta_z + Z \cdot a_{n-1}\T \Delta_z \cr
  &= y_{n-1} \COMMA
\end{align*}
i.e., the value of $w\T \Delta_z$ does not depend on the outcome of~$Z$ since~$\Delta_z$ is orthogonal to all~$a_k$. For~$\Delta'_z$ we obtain
\begin{align*}
  w\T \Delta'_z
  &= \big( M \cdot (y_1, \ldots, y_{n-1}, Z)\T \big)\T \Delta'_z \cr
  &= \sum_{k=1}^{n-2} y_k \cdot a_k\T \Delta'_z + y_{n-1} \cdot \Delta_z\T \Delta'_z + Z \cdot a_{n-1}\T \Delta'_z \cr
  &= y_{n-1} \cdot \Delta_z\T \Delta'_z + Z \cdot a_{n-1}\T \Delta'_z
\end{align*}
as~$\Delta'_z$ is orthogonal to all~$a_k$ except for $k = \ell = n-1$. The chain of equivalences
\begin{align*}
  &\left| \frac{w\T \Delta_z}{c\T \Delta_z} - \frac{w\T \Delta'_z}{c\T \Delta'_z} \right| \leq \eps \cr
  &\iff \frac{w\T \Delta'_z}{c\T \Delta'_z} \in [-\eps, \eps] + \frac{w\T \Delta_z}{c\T \Delta_z} \cr
  &\iff w\T \Delta'_z \in \Big[ -\eps \cdot |c\T \Delta'_z|, \eps \cdot |c\T \Delta'_z| \Big] + \frac{w\T \Delta_z}{c\T \Delta_z} \cdot c\T \Delta'_z \cr
  &\iff Z \cdot a_{n-1}\T \Delta'_z \in \Big[ -\eps \cdot |c\T \Delta'_z|, \eps \cdot |c\T \Delta'_z| \Big] + \frac{w\T \Delta_z}{c\T \Delta_z} \cdot c\T \Delta'_z - y_{n-1} \cdot \Delta_z\T \Delta'_z
\end{align*}
implies, that for event~$\F_\eps$ to occur~$Z$ must fall into an interval $I = I(y_1, \ldots, y_{n-1})$ of length $2\eps \cdot |c\T \Delta'_z|/|a_{n-1}\T \Delta'_z|$. The probability for this to happen is bounded from above by
\[
  \frac{2n \cdot 2\eps \cdot \frac{|c\T \Delta'_z|}{|a_{n-1}\T \Delta'_z|}}{\delta(r_1, \ldots, r_n) \cdot \min_{k \in [n]} \|r_k\|}
  = \underbrace{\frac{4n \cdot |c\T \Delta'_z|}{\delta(r_1, \ldots, r_n) \cdot \min_{k \in [n]} \|r_k\| \cdot |a_{n-1}\T \Delta'_z|}}_{\FED \gamma} \cdot \eps \COMMA
\]
where $[r_1, \ldots, r_n] = -M^{-1} \cdot [u_1, \ldots, u_n]$. This is due to $(Y_1, \ldots, Y_{n-1}, Z)\T = [r_1, \ldots, r_n] \cdot \lambda$ and Corollary~\ref{corollary.Prob:enough randomness} (applied with~$\phi=1$). Since the vectors $r_1, \ldots, r_n$ are linearly independent, $\delta(r_1, \ldots, r_n)$ is a well-defined positive value and $\min_{k \in [n]} \|r_k\| > 0$. Furthermore, $|a_{n-1}\T \Delta'_z| > 0$ since~$i_{n-1}$ is the constraint which is not tight for~$z_3$, but for~$z_2$. Hence, $\gamma < \infty$, and thus $\Pr{\left| \frac{w\T \Delta_z}{c\T \Delta_z} - \frac{w\T \Delta'_z}{c\T \Delta'_z} \right| \leq \eps} \to 0$ for $\eps \to 0$.

As there are at most $m^{3n}$ triples $(z_1, z_2, z_3)$ we have to consider, the claim follows by applying a union bound.
\end{proof}

\subsection{Proof of Lemma~\ref{lemma:reconstruct}}

\begin{proof}[Proof of Lemma~\ref{lemma:reconstruct}]
We consider a linear auxiliary function $\bar{w} \colon \RR^n \to \RR$, given by $\bar{w}(x) \DEF \tilde{w}\T x + \gamma \cdot b_i$. The paths $\bar{R} = R_{c, \bar{w}}$ and~$\tilde{R}$ are identical except for a shift by $\gamma \cdot b_i$ in the second coordinate because for $\bar{\pi} = \pi_{c, \bar{w}}$ we obtain
\[
  \bar{\pi}(x)
  = (c\T x, \tilde{w}\T x + \gamma \cdot b_i)
  = (c\T x, \tilde{w}\T x) + (0, \gamma \cdot b_i)
  = \tilde{\pi}(x) + (0, \gamma \cdot b_i)
\]
for all $x \in \RR^n$. Consequently, the slopes of~$\bar{R}$ and~$\tilde{R}$ are exactly the same (see Figure~\ref{fig:reconstruct shift}).

\begin{figure}
  \begin{center}
    \begin{subfigure}{0.4\textwidth}
      \begin{center}
        \includegraphics[page=1, width=0.9\textwidth]{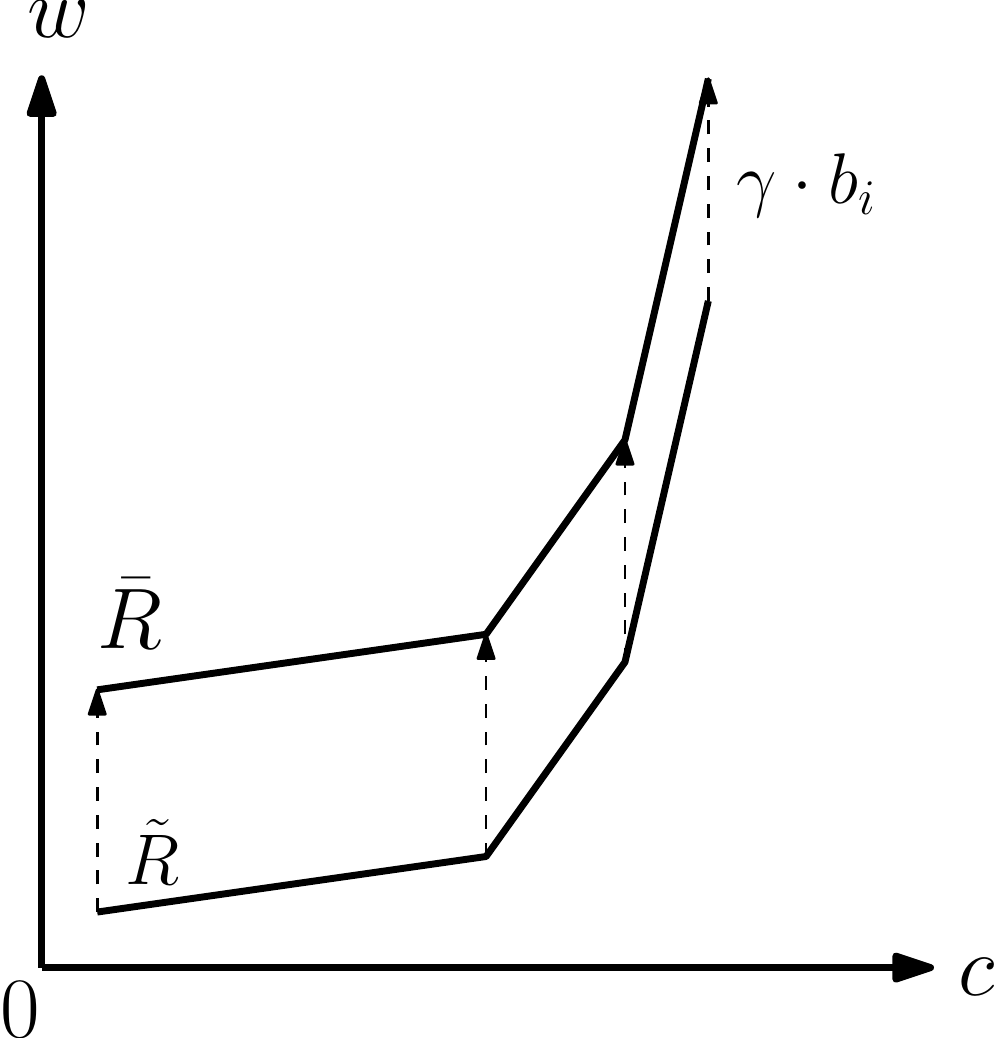}
        \caption{Relation between $\bar{R}$ and $\tilde{R}$}
        \label{fig:reconstruct shift}
      \end{center}
    \end{subfigure}
    \hspace{10ex}
    \begin{subfigure}{0.4\textwidth}
      \begin{center}
        \includegraphics[page=2, width=0.9\textwidth]{imgReconstruct.pdf}
        \caption{Relation between $\bar{R}$ an $R$}
        \label{fig:reconstruct below}
      \end{center}
    \end{subfigure}
  \end{center}
  \caption{Relations between $R$, $\tilde{R}$, and $\bar{R}$}
\end{figure}

Let $x \in P$ be an arbitrary point from the polytope~$P$. Then,
$
  \tilde{w}\T x
  = w\T x - \gamma \cdot a_i\T x
  \geq w\T x - \gamma \cdot b_i
$.
The inequality is due to $\gamma \geq 0$ and $a_i\T x \leq b_i$ for all $x \in P$. Equality holds, among others, for $x = x^\star$ due to the choice of~$a_i$. Hence, for all points $x \in P$ the two-dimensional points $\pi(x)$ and $\bar{\pi}(x)$ agree in the first coordinate while the second coordinate of $\pi(x)$ is at most the second coordinate of $\bar{\pi}(x)$ as $\bar{w}(x) = \tilde{w}\T x + \gamma \cdot b_i \geq w\T x$. Additionally, we have $\pi(x^\star) = \bar{\pi}(x^\star)$. Thus, path~$\bar{R}$ is above path~$R$ but they have point $p^\star = \pi(x^\star)$ in common. Hence, the slope of~$\bar{R}$ to the left (right) of~$p^\star$ is at most (at least) the slope of~$R$ to the left (right) of~$p^\star$ which is at most (greater than)~$t$ (see Figure~\ref{fig:reconstruct below}). Consequently, $p^\star$ is the rightmost vertex of~$\bar{R}$ whose slope does not exceed~$t$. Since~$\bar{R}$ and~$\tilde{R}$ are identical up to a shift of $(0, \gamma \cdot b_i)$, $\tilde{\pi}(x^\star)$ is the rightmost vertex of~$\tilde{R}$ whose slope does not exceed~$t$, i.e., $\tilde{\pi}(x^\star) = \tilde{p}^\star$.
\end{proof}

\subsection{Proof of Lemma~\ref{lemma:probability bound}}

\begin{proof}[Proof of Lemma~\ref{lemma:probability bound}]
Due to Lemma~\ref{lemma:event covering} it suffices to show that
\[
  \Pr{\E_{i, t, \eps}}
  \leq \frac{1}{m} \cdot \frac{2mn^2\eps}{\max \SET{ \frac n2, t } \cdot \delta^2}
  = \frac{2n^2\eps}{\max \SET{ \frac n2, t } \cdot \delta^2}
\]
for any index $i \in [m]$.

We apply the principle of deferred decisions and assume that vector~$c$ is already fixed. Now we extend the normalized vector~$a_i$ to an orthonormal basis $\SET{ q_1, \ldots, q_{n-1}, a_i }$ of~$\RR^n$ and consider the random vector $(Y_1, \ldots, Y_{n-1}, Z)\T = Q\T w$ given by the matrix vector product of the transpose of the orthogonal matrix $Q = [q_1, \ldots, q_{n-1}, a_i]$ and the vector $w = -[u_1, \ldots, u_n] \cdot \lambda$. For fixed values $y_1, \ldots, y_{n-1}$ let us consider all realizations of~$\lambda$ such that $(Y_1, \ldots, Y_{n-1}) = (y_1, \ldots, y_{n-1})$. Then,~$w$ is fixed up to the ray
\[
  w(Z)
  = Q \cdot (y_1, \ldots, y_{n-1}, Z)\T
  = \sum_{j=1}^{n-1} y_j \cdot q_j + Z \cdot a_i
  = v + Z \cdot a_i
\]
for $v = \sum_{j=1}^{n-1} y_j \cdot q_j$. All realizations of $w(Z)$ that are under consideration are mapped to the same value~$\tilde{w}$ by the function $w \mapsto \tilde{w}(w, i)$, i.e., $\tilde{w}(w(Z), i) = \tilde{w}$ for any possible realization of~$Z$. In other words, if $w = w(Z)$ is specified up to this ray, then the path $R_{c, \tilde{w}(w, i)}$ and, hence, the vectors~$y^\star$ and~$\hat{y}$ from the definition of event $\E_{i, t, \eps}$, are already determined.

Let us only consider the case that the first condition of event~$\E_{i, t, \eps}$ is fulfilled. Otherwise, event~$\E_{i, t, \eps}$ cannot occur. Thus, event $\E_{i, t, \eps}$ occurs iff
\[
  (t, t+\eps]
  \ni \frac{w\T \cdot (\hat{y} - y^\star)}{c\T \cdot (\hat{y} - y^\star)}
  = \underbrace{\frac{v\T \cdot (\hat{y} - y^\star)}{c\T \cdot (\hat{y} - y^\star)}}_{\FED \alpha} + Z \cdot \underbrace{\frac{a_i\T \cdot (\hat{y} - y^\star)}{c\T \cdot (\hat{y} - y^\star)}}_{\FED \beta} \DOT
\]
The next step in this proof will be to show that the inequality $|\beta| \geq \max \SET{ \frac n2, t } \cdot \frac{\delta}{n}$ is necessary for event~$\E_{i, t, \eps}$ to happen. For the sake of simplicity let us assume that $\|\hat{y} - y^\star\| = 1$ since~$\beta$ is invariant under scaling. If event~$\E_{i, t, \eps}$ occurs, then $a_i\T y^\star = b_i$, $\hat{y}$ is a neighbor of~$y^\star$, and $a_i\T \hat{y} \neq b_i$. That is, by Lemma~\ref{lemma:delta properties}, Claim~\ref{delta properties:neighboring vertices} we obtain $|a_i\T \cdot (\hat{y} - y^\star)| \geq \delta \cdot \|\hat{y} - y^\star\| = \delta$ and, hence,
\[
  |\beta|
  = \left| \frac{a_i\T \cdot (\hat{y} - y^\star)}{c\T \cdot (\hat{y} - y^\star)} \right|
  \geq \frac{\delta}{|c\T \cdot (\hat{y} - y^\star)|} \DOT
\]
On the one hand we have $|c\T \cdot (\hat{y} - y^\star)| \leq \|c\| \cdot \|\hat{y}-y^\star\| \leq \Big(1+\frac{\sqrt{n}}{\phi}\Big) \cdot 1 \le 2$,
where the second inequality is due to the choice of~$c$ as perturbation of the unit vector~$c_0$ and the third inequality is due to
the assumption~$\phi\ge\sqrt{n}$. 
On the other hand, due to $\frac{w\T \cdot (\hat{y} - y^\star)}{c\T \cdot (\hat{y} - y^\star)} \geq t$ we have
\[
  |c\T \cdot (\hat{y} - y^\star)|
  \leq \frac{|w\T \cdot (\hat{y} - y^\star)|}{t}
  \leq \frac{\|w\| \cdot \|\hat{y} - y^\star\|}{t}
  \leq \frac{n}{t} \DOT
\]
Consequently,
\[
  |\beta|
  \geq \frac{\delta}{\min \SET{ 2, \frac{n}{t} }}
  = \max \SET{ \frac n2, t } \cdot \frac{\delta}{n} \DOT
\]
Summarizing the previous observations we can state that if event~$\E_{i, t, \eps}$ occurs, then $|\beta| \geq \max \SET{ \frac n2, t } \cdot \frac{\delta}{n}$ and $\alpha + Z \cdot \beta \in (t, t+\eps]$. Hence,
\[
  Z \cdot \beta
  \in (t, t+\eps] - \alpha \COMMA
\]
i.e., $Z$ falls into an interval $I(y_1, \ldots, y_{n-1})$ of length at most $\eps/(\max \SET{ \frac n2, t } \cdot \delta/n) = n \eps/(\max \SET{ \frac n2, t } \cdot \delta)$ that only depends on the realizations $y_1, \ldots, y_{n-1}$ of $Y_1, \ldots, Y_{n-1}$. Let~$\B_{i, t, \eps}$ denote the event that~$Z$ falls into the interval $I(Y_1, \ldots, Y_{n-1})$. We showed that $\E_{i, t, \eps} \subseteq \B_{i, t, \eps}$. Consequently,
\[
  \Pr{\E_{i, t, \eps}}
  \leq \Pr{\B_{i, t, \eps}}
  \leq \frac{2n \cdot \frac{n \eps}{\max \sSET{ \frac n2, t } \cdot \delta}}{\delta(Q\T u_1, \ldots, Q\T u_n)}
  \leq \frac{2n^2\eps}{\max \sSET{ \frac n2, t } \cdot \delta^2} \COMMA
\]
where the second inequality is due to Corollary~\ref{corollary.Prob:enough randomness} (applied with~$\phi=1$): 
By definition, we have
\[
  (Y_1, \ldots, Y_{n-1}, Z)\T
  = Q\T w
  = Q\T \cdot -[u_1, \ldots, u_n] \cdot \lambda
  = [-Q\T u_1, \ldots, -Q\T u_n] \cdot \lambda \DOT
\]
The third inequality stems from the fact that
$
  \delta(-Q\T u_1, \ldots, -Q\T u_n)
  = \delta(u_1, \ldots, u_n)
  \geq \delta
$,
where the equality is due to the orthogonality of~$-Q$ (Claim~\ref{delta properties:orthogonal matrix} of Lemma~\ref{lemma:delta properties}).
\end{proof}

\section{Justification of Assumptions}\label{specialcases}

We assumed the matrix $A \in \RR^{m \times n}$ to have full column rank and we assumed the polyhedron $ \SET{ x \in \RR^n \WHERE Ax \leq b }$
to be bounded. In this section we show that this entails no loss of generality by giving transformations of arbitrary linear programs into linear programs
with full column rank whose polyhedra of feasible solutions are bounded.

\subsection{Raising the Rank of Matrix A} \label{dimension}

For the algorithm we have assumed that the matrix $A$ determining the polyhedron $P= \lbrace x \in \RR^n \WHERE Ax \leq b \rbrace$ has full column rank. In this section we provide a solution if this condition is not met. For this, we describe the transformation of $A$ into a matrix $A'$ with full column rank by adding new linearly independent rows (we will ensure that the $\delta$-distance property respectively the value of $\Delta$ is not violated by the transformation of~$A$ into~$A'$). 

\subsubsection{Transformation with respect to $\delta$}
Assume that we have an arbitrary matrix $A=[a_1,\dots,a_m]\T \in \RR^{m \times n}$ with rank $r=\rank(A)<n$. This implies that the polyhedron $P=\lbrace x\WHERE Ax\leq b \rbrace$ has no vertices. Let $c \in \RR^n$ be an arbitrary vector. Then the linear program $\max \lbrace c\T x \WHERE Ax \leq b \rbrace$ has either no solution (this is true if~$P$ is empty or~$c\T x$ is unbounded) or infinitely many solutions. We distinguish two different cases.

\textbf{Case 1: $c \in \SPAN{a_1, \dots, a_m}$}\\
Let $\SPAN{a_1, \dots, a_m}^{\perp}$ denote the orthogonal complement of $\SPAN{a_1, \dots, a_m}$. 
Furthermore let~$o_1, \dots, o_{n-r}$ be an orthonormal basis of $\SPAN{a_1, \dots, a_m}^{\perp}$. 
Then the set of solutions $\mathbb L= \lbrace \argmax c\T x \WHERE Ax \leq b\rbrace$ equals the set 
\[
   \tilde{\mathbb L}= \lbrace v+\argmax c\T x \WHERE Ax \leq b, \tilde{A}x=\mathbb O, v \in\SPAN{a_1, \dots, a_m}^{\perp} \rbrace,
\] 
where $\tilde A=[o_1,\dots, o_{n-r}]$. Thus we can add rows $[o_1,-o_1, \dots, o_{n-r}, -o_{n-r}]$ and extend the vector~$b$ by zero 
entries and calculate the set of solutions (note that the $\delta$-distance-property does not change under this extension of $A$ by 
Lemma \ref{newrows}). This equals the case where $n-r$ basis variables are known and we can proceed as in Section \ref{reduct} by reducing the polyhedron to dimension $r$. \\ 

\textbf{Case 2: $c \not\in \SPAN{a_1, \dots, a_m}$}\\
We maintain the notation from above. Then we have a linear combination 
\begin{align*}
c=\sum\limits_{i=1}^{r} \ell_i\cdot a_i +\sum\limits_{i=1}^{n-r} \tilde \ell_i\cdot o_i
\end{align*}
where $\tilde \ell_k \neq 0$ for at least one $k \in [n-r]$. Without loss of generality we may assume that $\tilde \ell_k > 0$. But $x$ is not bounded for direction $o_k$ by $A$ and thus $o_k$'s coefficient in the linear combination of $x$ may be chosen arbitrarily large. Thus $\max \lbrace c\T x \WHERE Ax \leq b \rbrace$ is unbounded.

Finally we prove that adding rows from the orthogonal complement of $\langle a_1, \dots, a_m \rangle$ to~$A$ does not change the $\delta$-distance property. 

\begin{lemma}\label{newrows}
Let $A=[a_1,\dots,a_m]\T \in \RR^{m \times n}$ be an arbitrary matrix of rank $r \leq n-1$ 
and~$\|a_i\|=1$ for $i\in [m]$. Let $v \in \RR^n$ be a vector such that $\|v\|=1$ and $\langle v, a_i \rangle=0$ for $i \in [m]$.
Then $\rank(A')=r+1$ and furthermore $\delta(A')=\delta(A)$ where $A'=[a_1,\dots,a_m,v]\T$ is defined by adding the new row $v\T$ to matrix $A$. 
\end{lemma}

\begin{proof}
First choose $r$ linearly independent rows of $A$. Without loss of generality we may assume $a_1, \dots, a_r$. 
To calculate $\delta(\lbrace a_1 \dots, a_{r-1} \rbrace, a_r)$ we choose a vertex $x \in \SPAN{a_1, \dots, a_r}$ 
with $x\cdot a_i=0$ for $i \in [r-1]$ and $x\cdot a_r=1$. Let $\alpha$ be the angle
between $a_r$ and $x$. Then $\delta(\lbrace a_1 \dots, a_{r-1} \rbrace, a_r)=\sin (\pi/2- \alpha)=\cos (\alpha)=\frac{\|a_r\|}{\|x\|}=\frac{1}{\|x\|}$.

Moreover let $v, o_2, \dots, o_{n-r}$ be an orthonormal basis of the orthogonal complement $\SPAN{a_1, \dots, a_r}^{\bot}$.
Then $x\cdot v=0$ and $x\cdot o_{i+1}=0$ for $i \in [n-r-1]$ because of $x \in \SPAN{a_1, \dots, a_r}$. Thus, $x$ is the unique solution of the 
system of linear equations
\begin{align*}
[a_1, \dots, a_r, v, o_2, \dots, o_{n-r}]\T x=e_r,
\end{align*}
where $e_r \in \RR^n$ denotes the $r^{\text{th}}$ canonical unit vector.

In accordance with Definition~\ref{definition:delta} and Lemma~\ref{lemma:delta properties} 1, 
we obtain $\delta(A)$ by choosing a solution with minimum norm over all such systems of linear equations and all vectors $e_i\in \RR^n$ with $i \in [r]$. 
Consider now the matrix $A'$ which is obtained by adding the row $v$ to $A$. To calculate $\delta (A')$ we have to calculate the minimal norm $\frac{1}{\|x\|}$ of the set of solutions of the systems of linear equations of the form 
\begin{align*}
[a'_1, \dots, a'_r, v, o_2, \dots, o_{n-r}]\T x=e_i,
\end{align*}
with $i \in [r+1]$. In the case where $i \leq r$ the set of systems of linear equations equals the set of systems of linear equation
from the case where we calculate $\delta(A)$. Thus the minimum norm does not change and we obtain $\delta(A')=\delta(A)$.

In the case where $i=r+1$ the solution of the systems of linear equations is given by $x=v$ and we obtain $1/\|x\|=1$. But this is the maximum norm, which can be reached by a solution $x$ and thus the minimum norm does not change at all which completes the proof.
\end{proof}

\subsubsection{Transformation with respect to $\Delta$}

If we want to ensure that the value $\Delta(A)$ does not change under the tranformation (which means $\Delta(A')=\Delta(A)$) we have to consider a slight modification of the above transformation. Especially, we will add vectors $e_i$ for $i \in [n]$ which are part of the canonical basis of~$\RR^n$ such that $e_i \not \in \SPAN{a_1, \dots, a_m}$. \\
Again, we know that in Case 1 the polyhedron is either empty or has infinitely many solutions. Thus, if we find a solution 
\[
   x'\in \lbrace \argmax c\T x \WHERE Ax \leq b, \tilde{A}x=\mathbb O \rbrace,
\] 
where $\tilde A=[e_{i_1},\dots, e_{i_{n-r}}]$ we already know that $x'$ also maximizes $c$ with respect to $P$. Furthermore if we are in Case 2 which means $c \not\in \SPAN{a_1, \dots, a_m}$ then the function $c \T x$ is unbounded for elements $x \in P$.
It remains to show that $\Delta$ does not change by adding rows $e_i$ to $A$.

\begin{lemma}
Let $A=[a_1,\dots,a_m]\T \in \ZZ^{m \times n}$ be an arbitrary matrix of rank $r \leq n-1$. 
Let $e_i \in \RR^n$ be a vector part of the canonical basis of $\RR^n$ such that $e_i \not \in \SPAN{a_1,\dots, a_m}$.
Then $\rank(A')=r+1$ and furthermore $\Delta(A')=\Delta(A)$ where $A'=[a_1,\dots,a_m,e_i]\T$ is defined by adding the new row $e_i \T$ to matrix $A$. 
\end{lemma}

\begin{proof}
Let $B$ be a submatrix of $A'$. Then either $B$ contains no entries from row $e_i \T$ (which means $\det(B) \leq \Delta(A)$) or one row of $B$ is a subvector $e'$ of $e_i \T$.
We distinguish two different cases: 

\textbf{Case 1:} $e'=0$. Then $B$ has a zero row and thus $\det(B)=0\le\Delta(A)$.

\textbf{Case 2:} $e'\neq 0$. In this case $B$ has a row ${e'}\T$, which is element of the canonical basis. Then $B$ has the form
\begin{align*}
\begin{pmatrix}
\tilde A \\
{e'}\T
\end{pmatrix}.
\end{align*}
Using the Laplace expansion, the absolute value of the determinant of $B$ is at most the absolute value of a determinant of a submatrix of $\tilde A$ which is a submatrix of $A$. We obtain $\det(B) \leq \Delta(A)$ which concludes the proof.
\end{proof}
\subsection{Translation into a Bounded Polyhedron} \label{unbounded}

\newcommand{\enc}{\textnormal{enc}}
\newcommand{\lcm}{\textnormal{lcm}}

For the algorithm we have assumed that the polyhedron $P= \lbrace x \in \RR \WHERE Ax \leq b \rbrace$ is bounded. This may be done because in the case where $P$ is unbounded we tranform~$P$ into a polytope~$P'$ and run the algorithm for~$P'$.  If the optimum solution is unique and not a vertex of $P$, then we assert that the linear program $\max \lbrace c\T x \WHERE Ax \leq b \rbrace$ is unbounded.
To transform $P$ we use the construction applied in \cite{GRE}. First we choose $n$ linearly independent rows of $A$. Without loss of generality we may assume the rows are given by $a_1, \dots, a_n$. If we find a ball $B(0)$ with radius $r$ which contains all vertices of $P$, then we define a parallelpiped
\begin{align*}
Z=\lbrace x \in \RR^n \WHERE -r \leq a_i \cdot x \leq r, i \in [n] \rbrace,
\end{align*}
which contains all vertices of $P$ and does not violate the $\delta$-distance property since it is defined by rows of $A$.
Finally, set $P'= P \cap Z$ and start the algorithm on polytope $P'$. Note that $P'$ has $\delta$-distance since the set of rows of $A$ did not change during the transformation.

To construct a ball with the desired properties we have to assume $A \in \mathbb Q^{m \times n}$ such as $b\in \mathbb Q^m$, which means no loss of generality for the implementation of the algorithm. By a slight generalization of Lemma~3.1.33 of \cite{GLS} all vertices of the polyhedron $P$ are contained in a ball $B(0)$ with radius $r=\sqrt{n}\cdot 2^{\enc(A,b)-n^2}\cdot \lcm(A)^n$ if $A \in \mathbb Q$, where the function $\enc$ returns the encoding length and the function $\lcm(M)$ for a rational matrix $M \in \mathbb{Q}^{m \times n}$ returns the least common multiple of the denominators of the entries of~$M$. As a convention, the denominator of~$0$ is defined as~$1$.
\begin{lemma}
If $P= \lbrace x \in \mathbb R^n \WHERE Ax \leq b \rbrace$ for  $A \in \mathbb Q^{m \times n}$ and $b \in \mathbb Q^m$, then all vertices of $P$ are contained in a ball $B$ around $0$ with radius $r=\sqrt{n}\cdot 2^{\enc(A,b)-n^2}\cdot \lcm(A)^n$.
\end{lemma} 

\begin{proof}
We have to calculate an upper bound for the length of all vertices. Thus, for each submatrix $B$ of $A$ of rank $n$ and the corresponding subvector $b'$ of $b$ we have to bound the length of the solution $x$ of $Bx=b'$. Applying Cramer's rule, the components $x_i$ of $x$ are given by 
\[
	x_i=\frac{\det (B_i)}{\det (B)} \COMMA
\]
where $B_i$ equals $B$ after replacing the $i\th$ column of $B$ by $b'$. We obtain $\lcm(B)^n \cdot \det (B)=\det(\lcm(B)\cdot B) \geq 1$ since $\lcm(B)\cdot B$ is integral and non-singular. All together we obtain 
\begin{align*}
x_i=\frac{\det (B_i)}{\det (B)} &\leq \det (B_i)\cdot \lcm(B)^n \\ &\leq 2^{\enc(B_i)-n^2}\cdot \lcm(B)^n \\ &\leq 2^{\enc(B,b')-n^2}\cdot \lcm(B)^n.
\end{align*}
Thus, choosing $r=\sqrt{n}\cdot 2^{\enc(A,b)-n^2}\cdot \lcm(A)^n$ the ball $B(0)$ with radius $r$ contains all vertices of $P$.
\end{proof}

\section{An Upper Bound on the Number of Random Bits}\label{sec:RandomBits}

For our analysis we assumed that we can draw continuous random variables. In practice it is, however, more realistic to assume that we can draw a finite number of random bits. In this section we will show that our algorithm only needs to draw $\text{poly}(\log m, n, \log(1/\delta))$ bits in order to obtain the expected running time stated in Theorem~\ref{thm:MainNumberOfPivots}. 
However, if the parameter~$\delta$ is not known to our algorithm, we have to
modify the shadow vertex algorithm. This will give us an additional factor of $O(n)$ in the expected running time.

Let us assume that we want to approximate a uniform random draw~$X$ from the interval $[0, 1)$ with~$k$ random bits $Y_1, \ldots, Y_k \in \SET{ 0, 1 }$. (A draw from an arbitrary interval $[a, b)$ can be simulated by drawing a random variable from $[0, 1)$ and then applying the affine linear function $x \mapsto a + (b-a) \cdot x$.) We consider the random variable $Z = \sum_{\ell=1}^k Y_\ell \cdot 2^{-\ell}$. We observe that the random variable~$Z$ has the same distribution as the random variable $g(X)$, where $g(x) = \lfloor x \cdot 2^k \rfloor/2^k$. Note that $|g(X)-X| \leq 2^{-k}$. Hence, instead of considering discrete variables and going through the whole analysis again, we will argue that, with high probability, the number of slopes of the shadow vertex polygon does not change if each random variable is perturbed by not more than a sufficiently small~$\eps$. If we have proven such a statement, this implies that we can approximate our continuous uniform random draws as discussed above by using $O \big( \log(1/\eps) \big)$ bits for each draw. Recall that our algorithm draws two random vectors $\lambda \in (0, 1]^n$ and $c \in [-1, 1]^n$ that we have to deal with in this section.

For a vector $x \in \RR^n$ and a real $\eps > 0$ let $U_\eps(x) \subseteq [-1,1]^n$ denote the set of vectors~$x'\in[-1,1]^n$ for which $\|x'-x\|_\infty \leq \eps$, that is, $x'$ and~$x$ differ in each component by at most~$\eps$. In the remainder let us only consider values $\eps \in (0, 1]$.

Whenever a vector $c \in [-1,1]^n$ and a vector $\hat{c} \in U_\eps(c)$ are defined, then by~$\Delta_c$ we refer to the difference $\Delta_c \DEF \hat{c} - c$. Observe that $\|\Delta_c\| \leq \sqrt{n} \eps$. The same holds for the vectors $\lambda \in (0, 1]^n$, $\hat{\lambda} \in U_\eps(\lambda)$, and $\Delta_\lambda \DEF \hat{\lambda} - \lambda$. When the vectors~$\lambda$ and~$\hat{\lambda}$ are defined, then the vectors~$w$ and~$\hat{w}$ are defined as $w \DEF -[u_1, \ldots, u_n] \cdot \lambda$ and $\hat{w} \DEF -[u_1, \ldots, u_n] \cdot \hat{\lambda}$ (cf.\ Algorithm~\ref{algorithm:SV}). Furthermore, the vector~$\Delta_w$ is defined as $\Delta_w \DEF \hat{w} - w$. Note that $\|w\| = \|[u_1, \ldots, u_n] \cdot \lambda\| \leq \sum_{\ell=1}^n \|u_\ell\| \leq n$ as the rows $u_1\T, \ldots, u_n\T$ of matrix~$A$ are normalized. Similarly, $\|\hat{w}\| \leq n$ and $\|\Delta_w\| \leq n \eps$. We will frequently make use of these inequalities without discussing their correctness again.

If~$P$ denotes the non-degenerate bounded polyhedron $\SET{ x \in \RR^n \WHERE Ax \leq b }$, then we denote by $V_k(P)$ the set of all $k$-tuples $(z_1, \ldots, z_k)$ of pairwise distinct vertices $z_1, \ldots z_k$ of~$P$ such that for any $i = 1, \ldots, k-1$ the vertices~$z_i$ and~$z_{i+1}$ are neighbors, that is, they share exactly $n-1$ tight constraints. In other words, $V_k(P)$ contains the set of all simple paths of length $k-1$ of the edge graph of~$P$. Note that $|V_k(P)| \leq \binom{m}{n} \cdot n^{k-1} \leq m^n n^{k-2}$. For our analysis only $V_2(P)$ and $V_3(P)$ are relevant.

The following lemma is an adaption of Lemma~\ref{lemma:failure probability I} for our needs in this section and follows from Lemma~\ref{lemma:failure probability I}.

\begin{lemma}
\label{lemma:swap c}
The probability that there exist a pair $(z_1, z_2) \in V_2(P)$ and a vector $\hat{c} \in U_\eps(c)$ for which $\hat{c}\T \cdot (z_2-z_1) = 0$ is bounded from above by $2m^n n^{3/2} \eps \phi$.
\end{lemma}

\begin{proof}
Let $c \in [-1,1]^n$ be a vector such that there exists a vector $\hat{c} \in U_\eps(c)$ for which $\hat{c}\T \cdot (z_2-z_1) = 0$ for an appropriate pair $(z_1, z_2) \in V_2(P)$. Then
\begin{align*}
	|c\T \cdot (z_2-z_1)|
	&= |\hat{c}\T \cdot (z_2-z_1) - \Delta_c\T \cdot (z_2-z_1)| \cr
	&\leq \|\Delta_c\| \cdot \|z_2-z_1\| \cr
	&\leq \sqrt{n} \eps \cdot \|z_2-z_1\| \DOT
\end{align*}
In accordance with Lemma~\ref{lemma:failure probability I}, the probability of this event is bounded from above by $2m^n n^{3/2} \eps \phi$.
\end{proof}

A similar statement as Lemma~\ref{lemma:swap c} can be made for the objective~$w$. However, for our purpose we need a slightly stronger statement.

\begin{lemma}
\label{lemma:swap w}
The probability that there exist a pair $(z_1, z_2) \in V_2(P)$ and a vector $\hat{\lambda} \in U_\eps(\lambda)$ for which $|\hat{w}\T \cdot (z_2-z_1)| \leq n \eps^{1/3} \cdot \|z_2-z_1\|$, where $\hat{w} = -[u_1, \ldots, u_n] \cdot \hat{\lambda}$ (cf.\ Algorithm~\ref{algorithm:SV}), is bounded from above by $4m^n n^2 \eps^{1/3}/\delta$.
\end{lemma}

\begin{proof}
Fix a pair $(z_1, z_2) \in V_2(P)$ and let $\Delta_z \DEF z_2-z_1$. Without loss of generality let us assume that $\|\Delta_z\| = 1$. The event $\hat{w}\T \Delta_z \in [-n \eps^{1/3}, n \eps^{1/3}]$ is equivalent to
\[
	w\T \Delta_z \in [-n \eps^{1/3}, n \eps^{1/3}] - \Delta_w\T \Delta_z \DOT
\]
This interval is a subinterval of $[-2n \eps^{1/3}, 2n \eps^{1/3}]$ as
\[
	|\Delta_w\T \Delta_z|
	\leq \|\Delta_w\| \cdot \|\Delta_z\|
	\leq n \eps \cdot 1
	\leq n \eps^{1/3}
\]
when recalling that $\eps \leq 1$. Since
\begin{align*}
	w\T \Delta_z \in [-2n \eps^{1/3}, 2n \eps^{1/3}]
	&\iff (U \lambda)\T \Delta_z \in [-2n \eps^{1/3}, 2n \eps^{1/3}] \cr
	&\iff \lambda\T y \in [-2n \eps^{1/3}, 2n \eps^{1/3}]
\end{align*}
for $U = [u_1, \ldots, u_n]$ and $y = U\T \Delta_z$, in the next part of this proof we will derive a lower bound for $\|y\|$. Particularly, we will show that $\|y\| \geq \delta/\sqrt{n}$.

Let $M \DEF [m_1, \ldots, m_n] \DEF (U\T)^{-1}$. Due to $\Delta_z = My$, we obtain $1 = \|\Delta_z\| \leq \|M\| \cdot \|y\|$, which implies $\|y\| \geq 1/\|M\|$. In accordance with Lemma~\ref{lemma:delta properties}, Claim~\ref{delta properties:inverse}, we obtain
\[
	\max_{k \in [n]} \|m_k\|
	= \frac{1}{\delta(u_1, \ldots, u_n)}
	\leq \frac{1}{\delta} \DOT
\]
Consequently,
\[
	\|Mx\|
	\leq \sum_{k=1}^n \|m_k\| \cdot |x_k|
	\leq \sum_{k=1}^n \frac{1}{\delta} \cdot |x_k|
	= \frac{\|x\|_1}{\delta}
	\leq \frac{\sqrt{n} \cdot \|x\|}{\delta}
\]
for any vector $x \neq \NULL$, i.e., $\|M\| = \sup_{x \neq \NULL} \|Mx\|/\|x\| \leq \sqrt{n}/\delta$. Summarizing the previous observations, we obtain $\|y\| \geq 1/\|M\| \geq \delta/\sqrt{n}$.

For the last part of the proof we observe that there exists an index $i \in [n]$ such that $|y_i| \geq \delta/n$. We apply the principle of deferred decisions an assume that all coefficients $\lambda_j$ for $j \neq i$ are fixed arbitrarily. By the chain of equivalences
\begin{align*}
			&\lambda\T y \in [-2n \eps^{1/3}, 2n \eps^{1/3}] \cr
\iff	&\sum_{k=1}^n \frac{\lambda_k \cdot y_k}{y_i} \in \left[ -\frac{2n \eps^{1/3}}{|y_i|}, \frac{2n \eps^{1/3}}{|y_i|} \right] \cr
\iff  &\lambda_i \in \left[ -\frac{2n \eps^{1/3}}{|y_i|}, \frac{2n \eps^{1/3}}{|y_i|} \right] - \sum_{k \neq i} \frac{\lambda_k \cdot y_k}{y_i}
\end{align*}
we see that the event $\lambda\T y \in [-2n \eps^{1/3}, 2n \eps^{1/3}]$ occurs if and only if the coefficient~$\lambda_i$, which we did not fix, falls into a certain fixed interval of length $4n \eps^{1/3}/|y_i|$. The probability for this to happen is at most $4n \eps^{1/3}/|y_i| \leq 4n^2 \eps^{1/3}/\delta$. The claim follows by applying a union bound over all pairs $(z_1, z_2) \in V_2(P)$, which gives us the additional factor of $m^n$.
\end{proof}

The next observation characterizes the situation when the projections of two linearly independent vectors in~$\RR^n$ are projected onto two linearly dependent vectors in~$\RR^2$ by the function $x \mapsto (\hat{c}\T x, \hat{w}\T x)$.

\begin{observation}
\label{observation:same slopes}
Let $(z_1, z_2, z_3) \in V_3(P)$, let $\Delta_1 \DEF z_2-z_1$ and $\Delta_2 \DEF z_3-z_2$, and let $\hat{c}, \hat{w} \in \RR^n$ be vectors for which $\hat{w}\T \Delta_1 \neq 0$, $\hat{w}\T \Delta_2 \neq 0$, and
\[
	\frac{\hat{w}\T \Delta_1}{\hat{c}\T \Delta_1}
	= \frac{\hat{w}\T \Delta_2}{\hat{c}\T \Delta_2} \DOT
\]
Then $\hat{c}\T x = 0$ for $x \DEF \Delta_1 - \mu \cdot \Delta_2$, where $\mu = \hat{w}\T \Delta_1/\hat{w}\T \Delta_2$.
\end{observation}

Note that, by the definition of~$x$, the equation $\hat{w}\T x = 0$ trivially holds. For the equation $\hat{c}\T x = 0$ we require that the projections of~$\Delta_1$ and~$\Delta_2$ are linearly dependent as it is assumed in Observation~\ref{observation:same slopes}. Furthermore, let us remark that in the formulation above we allow $\hat{c}\T \Delta_1 = 0$ or $\hat{c}\T \Delta_2 = 0$ using the convention $x/0 = +\infty$ for $x > 0$ and $x/0 = -\infty$ for $x < 0$.

\begin{proof}
The claim follows from
\begin{align*}
	\hat{c}\T x
	&= \hat{c}\T \Delta_1 - \mu \cdot \hat{c}\T \Delta_2
	= \hat{c}\T \Delta_2 \cdot \frac{\hat{w}\T \Delta_1}{\hat{w}\T \Delta_2} - \mu \cdot \hat{c}\T \Delta_2 \cr
	&= \hat{c}\T \Delta_2 \cdot \frac{\mu \cdot \hat{w}\T \Delta_2}{\hat{w}\T \Delta_2} - \mu \cdot \hat{c}\T \Delta_2
	= 0 \DOT \qedhere
\end{align*}
\end{proof}

We are now able to prove an analog of Lemma~\ref{lemma:failure probability II}.

\begin{lemma}
\label{lemma:same slopes}
The probability that there exist a triple $(z_1, z_2, z_3) \in V_3(P)$ and vectors $\hat{\lambda} \in U_\eps(\lambda)$ and $\hat{c} \in U_\eps(c)$ for which
\[
	\frac{\hat{w}\T \Delta_1}{\hat{c}\T \Delta_1}
	= \frac{\hat{w}\T \Delta_2}{\hat{c}\T \Delta_2} \COMMA
\]
where $\Delta_1 \DEF z_2-z_1$, $\Delta_2 \DEF z_3-z_2$, and $\hat{w} = -[u_1, \ldots, u_n] \cdot \hat{\lambda}$, is bounded from above by $12m^n n^2 \eps^{1/3} \phi/\delta$.
\end{lemma}

\begin{proof}
Let us introduce the following events:
\begin{itemize}[leftmargin=0.8cm]

	\item With event~$A$ we refer to the event stated in Lemma~\ref{lemma:same slopes}.
	
	\item Event~$B$ occurs if there exist a pair $(z_1, z_2) \in V_2(P)$ and a vector $\hat{\lambda} \in U_\eps(\lambda)$ such that $|\hat{w}\T \cdot (z_2-z_1)| \leq n \eps^{1/3} \cdot \|z_2-z_1\|$ (cf.\ Lemma~\ref{lemma:swap w}).
	
	\item Event~$C$ occurs if there is a triple $(z_1, z_2, z_3) \in V_3(P)$ such that $|c\T x| \leq (4\sqrt{n} \eps^{1/3}/\delta) \cdot \|x\|$, where $x = x(w, z_1, z_2, z_3) \DEF \Delta_1 - \mu \cdot \Delta_2$ for $\Delta_1 \DEF z_2-z_1$, $\Delta_2 \DEF z_3-z_2$, and $\mu = w\T \Delta_1/w\T \Delta_2$ if $w\T \Delta_2 \neq 0$ and $\mu = 0$ otherwise (cf.\ Observation~\ref{observation:same slopes}).
	
\end{itemize}
In the first part of the proof we will show that $A \subseteq B \cup C$. For this, it suffices to show that $A \setminus B \subseteq C$. Let us consider realizations $w \in (0,1]^n$ and $c \in [-1,1]^n$ for which event~$A$ occurs, but not event~$B$. Let $(z_1, z_2, z_3) \in V_3(P)$, $\hat{\lambda} \in U_\eps(\lambda)$, and $\hat{c} \in U_e(c)$ be the vectors mentioned in the definition of event~$A$. Our goal is to show that $|c\T x| \leq (4\sqrt{n} \eps^{1/3}/\delta) \cdot \|x\|$ for $x = x(w, z_1, z_2, z_3)$. As event~$B$ does not occur, we know that
\begin{align*}
	&|w\T \Delta_1| \geq n \eps^{1/3} \cdot \|\Delta_1\| \COMMA \quad \hphantom{\text{and}} \quad
	|\hat{w}\T \Delta_1| \geq n \eps^{1/3} \cdot \|\Delta_2\| \COMMA \cr
	&|w\T \Delta_2| \geq n \eps^{1/3} \cdot \|\Delta_2\| \COMMA \quad \text{and} \quad
	|\hat{w}\T \Delta_2| \geq n \eps^{1/3} \cdot \|\Delta_2\| \DOT
\end{align*}
Furthermore, note that
\[
	|\hat{w}\T \Delta_1 - w\T \Delta_1|
	\leq \|\Delta_w\| \cdot \|\Delta_1\|
	\leq n \eps \cdot \|\Delta_1\|
\]
and, similarly,
\[
	|\hat{w}\T \Delta_2 - w\T \Delta_2| \leq n \eps \cdot \|\Delta_2\| \DOT
\]
Therefore,
\begin{align*}
	|\hat{w}\T \Delta_1 - w\T \Delta_1|
	&\leq n\eps \cdot \|\Delta_1\|
		\leq \eps^{2/3} \cdot |w\T \Delta_1| \quad \text{and} \cr
	|\hat{w}\T \Delta_2- w\T \Delta_2|
	&\leq n\eps \cdot \|\Delta_2\|
	\leq \eps^{2/3} \cdot |\hat{w}\T \Delta_2| \COMMA
\end{align*}
and, consequently
\begin{align*}
	\frac{|\hat{w}\T \Delta_1|}{|\hat{w}\T \Delta_2|}
	&\leq \frac{(1+\eps^{2/3}) \cdot |w\T \Delta_1|}{\frac{1}{1+\eps^{2/3}} \cdot |w\T \Delta_2|}
		= (1+\eps^{2/3})^2 \cdot \frac{|w\T \Delta_1|}{|w\T \Delta_2|}
		\leq (1+3\eps^{2/3}) \cdot \frac{|w\T \Delta_1|}{|w\T \Delta_2|} \quad \text{and} \cr
	\frac{|\hat{w}\T \Delta_1|}{|\hat{w}\T \Delta_2|}
	&\geq \frac{(1-\eps^{2/3}) \cdot |w\T \Delta_1|}{\frac{1}{1-\eps^{2/3}} \cdot |w\T \Delta_2|}
		= (1-\eps^{2/3})^2 \cdot \frac{|w\T \Delta_1|}{|w\T \Delta_2|}
		\geq (1-3\eps^{2/3}) \cdot \frac{|w\T \Delta_1|}{|w\T \Delta_2|} \DOT
\end{align*}
Here we again used $\eps \leq 1$. Observe that both, $\hat{w}\T \Delta_1$ and $w\T \Delta_1$, as well as $\hat{w}\T \Delta_2$ and $w\T \Delta_2$, have the same sign, since their absolute values are larger than $n \eps^{1/3} \cdot \|\Delta_1\|$ and $n \eps^{1/3} \cdot \|\Delta_2\|$, but their difference is at most $n \eps \cdot \|\Delta_1\|$ and $n \eps \|\Delta_2\|$, respectively. Hence,
\[
	\left| \frac{\hat{w}\T \Delta_1}{\hat{w}\T \Delta_2} - \frac{w\T \Delta_1}{w\T \Delta_2} \right|
	= \left| \left| \frac{\hat{w}\T \Delta_1}{\hat{w}\T \Delta_2} \right| - \left| \frac{w\T \Delta_1}{w\T \Delta_2} \right| \right|
	\leq 3\eps^{2/3} \cdot \frac{|w\T \Delta_1|}{|w\T \Delta_2|} \DOT
\]
As event~$A$ occurs, but not event~$B$, Observation~\ref{observation:same slopes} yields $\hat{c}\T x(\hat{w}, z_1, z_2, z_3) = 0$. With the previous inequality we obtain
\begin{align*}
	|\hat{c}\T x(w, z_1, z_2, z_3)|
	&= \left| \hat{c}\T \cdot \big( x(w, z_1, z_2, z_3) - x(\hat{w}, z_1, z_2, z_3) \big) \right| \cr
	&\leq \|\hat{c}\| \cdot \|x(w, z_1, z_2, z_3) - x(\hat{w}, z_1, z_2, z_3)\| \cr
	&= \|\hat{c}\| \cdot \left| \frac{w\T \Delta_1}{w\T \Delta_2} - \frac{\hat{w}\T \Delta_1}{\hat{w}\T \Delta_2} \right| \cdot \|\Delta_2\| \cr
	&\leq \sqrt{n} \cdot 3\eps^{2/3} \cdot \frac{|w\T \Delta_1|}{|w\T \Delta_2|} \cdot \|\Delta_2\| \cr
	&\leq \sqrt{n} \cdot 3\eps^{2/3} \cdot \frac{\|w\| \cdot \|\Delta_1\|}{n \eps^{1/3} \cdot \|\Delta_2\|} \cdot \|\Delta_2\| \cr
	&\leq \sqrt{n} \cdot 3\eps^{2/3} \cdot \frac{n \cdot \|\Delta_1\|}{n \eps^{1/3} \cdot \|\Delta_2\|} \cdot \|\Delta_2\| \cr
	&= 3\sqrt{n} \eps^{1/3} \cdot \|\Delta_1\| \DOT
\end{align*}
In the remainder of this proof, with~$x$ we refer to the vector $x(w, z_1, z_2, z_3)$ (and not to, e.g., $x(\hat{w}, z_1, z_2, z_3)$). Now we show that $\|x\| \geq \delta \cdot \|\Delta_1\|$. For this, let~$a_i\T$ be a row of matrix~$A$ for which $a_i\T z_1 < b_i$, but $a_i\T z_2 = a_i\T z_3 = b_i$, i.e., the $i\th$ constraint is tight for~$z_2$ and~$z_3$, but not for~$z_1$. Such a constraint exists as~$z_1$ and~$z_3$ are distinct neighbors of~$z_2$. Consequently, $a_i\T \Delta_1 > 0$ and $a_i\T \Delta_2 = 0$. Hence,
\[
	|a_i\T x|
	= |a_i\T \cdot (\Delta_1 - \mu \cdot \Delta_2)|
	= |a_i\T \cdot \Delta_1|
	\geq \delta \cdot \|\Delta_1\| \COMMA
\]
where the last inequality is due to Lemma~\ref{lemma:delta properties}, Claim~\ref{delta properties:neighboring vertices}. As $\|a_i\| = 1$, we obtain
\[
	\|x\|
	\geq \frac{|a_i\T x|}{\|a_i\|}
	= |a_i\T x|
	\geq \delta \cdot \|\Delta_1\| \DOT
\]
Summarizing the previous observations yields
\[
	|\hat{c}\T x|
	\leq 3\sqrt{n} \eps^{1/3} \cdot \|\Delta_1\|
	\leq \frac{3\sqrt{n} \eps^{1/3}}{\delta} \cdot \|x\| \DOT
\]
Now that we have bounded $|\hat{c}\T x|$ from above, we easily get an upper bound for $|c\T x|$. Since
\[
	|c\T x - \hat{c}\T x|
	\leq \|\Delta_c\| \cdot \|x\|
	\leq \sqrt{n} \eps \cdot \|x\| \COMMA
\]
we obtain
\[
	|c\T x|
	\leq |\hat{c}\T x| + |c\T x - \hat{c}\T x|
	\leq \frac{3\sqrt{n} \eps^{1/3}}{\delta} \cdot \|x\| + \sqrt{n} \eps \cdot \|x\|
	\leq \frac{4\sqrt{n} \eps^{1/3}}{\delta} \cdot \|x\| \COMMA
\]
i.e., event~$C$ occurs.

In the second part of the proof we show that $\Pr{C} \leq 8m^n n^2 \eps^{1/3} \phi/\delta$. Due to $A \subseteq B \cup C$, $\phi \geq 1$, and Lemma~\ref{lemma:swap w}, it then follows that
\[
	\Pr{A}
	\leq 4m^n n^2 \eps^{1/3}/\delta + 8m^n n^2 \eps^{1/3} \phi/\delta
	\leq 12m^n n^2 \eps^{1/3} \phi/\delta \DOT
\]
Let $(z_1, z_2, z_3) \in V_3(P)$ be a triple of vertices of~$P$. We apply the principle of deferred decisions twice: First, we assume that~$\lambda$ has already been fixed arbitrarily. Hence, the vector $x = x(w, z_1, z_2, z_3) \neq \NULL$ is also fixed. Let $z = (1/\|x\|) \cdot x$ be the normalization of~$x$. As $|c\T x| \leq (4\sqrt{n} \eps^{1/3}/\delta) \cdot \|x\|$ holds if and only if $|c\T z| \leq 4\sqrt{n} \eps^{1/3}/\delta$, we will analyze the probability of the latter event.

There exists an index~$i$ such that $|z_i| \geq 1/\sqrt{n}$. Now we again apply the principle of deferred decisions an assume that all coefficients~$c_j$ for $j \neq i$ are fixed arbitrarily. Then
\begin{align*}
	|c\T z| \leq 4\sqrt{n} \eps^{1/3}/\delta
	&\iff \sum_{j=1}^n c_j \cdot \frac{z_j}{z_i} \in \left[ -\frac{4\sqrt{n} \eps^{1/3}}{\delta \cdot |z_i|}, \frac{4\sqrt{n} \eps^{1/3}}{\delta \cdot |z_i|} \right] \cr
	&\iff c_i \in \left[ -\frac{4\sqrt{n} \eps^{1/3}}{\delta \cdot |z_i|}, \frac{4\sqrt{n} \eps^{1/3}}{\delta \cdot |z_i|} \right] - \sum_{j \neq i} c_j \cdot \frac{z_j}{z_i} \DOT
\end{align*}
Hence, the random coefficient~$c_i$ must fall into a fixed interval of length $8\sqrt{n} \eps^{1/3}/(\delta \cdot |z_i|)$. The probability for this to happen is at most
\[
	\frac{8\sqrt{n} \eps^{1/3}}{\delta \cdot |z_i|} \cdot \phi
	\leq \frac{8\sqrt{n} \eps^{1/3}}{\delta \cdot \frac{1}{\sqrt{n}}} \cdot \phi
	= \frac{8n \eps^{1/3} \phi}{\delta} \DOT
\]
A union bound over all triples $(z_1, z_2, z_3) \in V_3(P)$ gives the additional factor of $V_3(P) \leq m^n n$.
\end{proof}

\begin{lemma}
\label{lemma:enough random bits}
Let us consider the shadow vertex algorithm given as Algorithm~\ref{algorithm:SV} for $\phi \geq \sqrt{n}$. If we replace the draw of each continuous random variable by the draw of at least
\[
	B(m, n, \phi, \delta) \DEF \ceil{6n \log_2 m + 6 \log_2 n + 3 \log_2 \phi + 3 \log_2(1/\delta) + 12}
\]
random bits as described earlier in this section, then the expected number of pivots is $O \big( \frac{mn^2}{\delta^2} + \frac{m \sqrt{n} \phi}{\delta} \big)$.
\end{lemma}

\begin{proof}
As discussed in the beginning of this section, instead of drawing~$k$ random bits to simulate a uniform random draw from an interval $[a, b)$, we can draw a uniform random variable~$X$ from $[0, 1)$ and apply the function $g(X) = h(\lfloor X \cdot 2^k \rfloor/2^k)$ for $h(x) = a + (b-a) \cdot x$ to obtain a discrete random variable with the same distribution. Observe, that $|X-g(X)| \leq (b-a)/2^k$. In the shadow vertex algorithm all intervals are of length~$1$ or of length $1/\phi \leq 1$. Hence, $|X-g(X)| \leq 2^{-k}$. As we use $k \geq B(m, n, \phi, \delta)$ bits for each draw, we obtain $g(X) \in U_\eps(X)$ for
\[
	\eps
	= 2^{-B(m, n, \phi, \delta)}
	\leq \frac{\delta^3}{2^{12} m^{6n} n^6 \phi^3}
	= \left( \frac{\delta}{16 m^{2n} n^2 \phi} \right)^3 \DOT
\]
Now let~$c$ and~$\lambda$ denote the continuous random vectors and let $\bar{c} \in U_\eps(c)$ and $\bar{\lambda} \in U_\eps(\lambda)$ denote the discrete random vectors obtained from~$c$ and~$\lambda$ as described above. Furthermore, let $w = -[u_1, \ldots, u_n] \cdot \lambda$ and $\bar{w} = -[u_1, \ldots, u_n] \cdot \bar{\lambda}$. We introduce the event~$D$ which occurs if one of the following holds:
\begin{enumerate}

	\item There exists a pair $(z_1, z_2) \in V_2(P)$ such that $c\T z_1$ and $c\T z_2$ are not in the same relation as $\bar{c}\T z_1$ and $\bar{c}\T z_2$ or $c\T z_1 = c\T z_2$ or $\bar{c}\T z_1 = \bar{c}\T z_2$.
	
	\item There exists a triple $(z_1, z_2, z_3) \in V_3(P)$ such that $\frac{w\T \cdot (z_2-z_1)}{c\T \cdot (z_2-z_1)}$ and $\frac{w\T \cdot (z_3-z_2)}{c\T \cdot (z_3-z_2)}$ are not in the same relation as $\frac{\bar{w}\T \cdot (z_2-z_1)}{\bar{c}\T \cdot (z_2-z_1)}$ and $\frac{\bar{w}\T \cdot (z_3-z_2)}{\bar{c}\T \cdot (z_3-z_2)}$.

\end{enumerate}
\newcommand{\sgn}{\text{sgn}}
Here, $a$ and~$b$ being in the same relation as~$\bar{a}$ and~$\bar{b}$ means that $\sgn(a-b) = \sgn(\bar{a}-\bar{b})$, where $\sgn(x) = -1$ for $x < 0$, $\sgn(x) = 0$ for $x = 0$, and $\sgn(x) = +1$ for $x > 0$.

Let~$X$ and~$\bar{X}$ denote the number of pivots of the shadow vertex algorithm with continuous random vectors~$c$ and~$\lambda$ and with discrete random vectors~$\bar{c}$ and~$\bar{\lambda}$, respectively. We will first argue that $X = \bar{X}$ if event~$D$ does not occur. In both cases, we start in the same vertex~$x_0$. In each vertex~$x$, the algorithm chooses among the neighbors of~$x$ with a larger $c$-value (or $\bar{c}$-value, respectively) the neighbor~$z$ with the smallest slope $\frac{w\T \cdot (z-x)}{c\T \cdot (z-x)}$ (or $\frac{\bar{w}\T \cdot (z-x)}{\bar{c}\T \cdot (z-x)}$, respectively). If event~$D$ does not occur, then in both cases the same neighbors of~$x$ are considered and, additionally, the order of their slopes is the same. Hence, in both cases the same sequence of vertices is considered.

Now let~$Y$ be the random variable that takes the value $m^n$ if event~$D$ occurs and the value~$0$ otherwise. Clearly, $\bar{X} \leq X + Y$ and, thus, 
\[
	\Ex{\bar{X}}
	\leq \Ex{X} + \Ex{Y}
	\leq O \left( \frac{mn^2}{\delta^2} + \frac{m \sqrt{n} \phi}{\delta} \right) + m^n \cdot \Pr{D} \COMMA
\]
where the last inequality stems from Theorem~\ref{main}. In the remainder of this proof we show that the probability $\Pr{D}$ of event~$D$ is bounded from above by $1/m^n$. For this, let us assume that the first part of the definition of event~$D$ is fulfilled for a pair $(z_1, z_2) \in V_2(P)$. If $c\T z_1$ and $c\T z_2$ are not in the same relation as $\bar{c}\T z_1$ and $\bar{c}\T z_2$, then there exists a $\mu \in [0, 1]$ such that
\[
	\mu \cdot (c\T z_1 - c\T z_2) + (1-\mu) \cdot (\bar{c}\T z_1 - \bar{c}\T z_2) = 0 \DOT
\]
If we consider the vector $\hat{c} \DEF \mu \cdot c + (1-\mu) \cdot \bar{c} \in U_\eps(c)$, then we obtain
\[
	\hat{c}\T \cdot (z_2 - z_1)
	= \mu \cdot c\T \cdot (z_2 - z_1) + (1-\mu) \cdot \bar{c}\T \cdot (z_2-z_1)
	= 0 \DOT
\]
Hence, the event described in Lemma~\ref{lemma:swap c} occurs. This event also occurs if $c\T z_1 = c\T z_2$ or $\bar{c}\T z_1 = \bar{c}\T z_2$.

Let us now assume that the second part of the definition of event~$D$ is fulfilled for a triple $(z_1, z_2, z_3) \in V_3(P)$, but not the first one, and let us consider the function $f \colon [0, 1] \to \RR$, defined by
\[
	f(\mu)
	= \frac{\big( \mu \cdot w + (1-\mu) \cdot \bar{w} \big)\T \cdot (z_2-z_1)}{\big( \mu \cdot c + (1-\mu) \cdot \bar{c} \big)\T \cdot (z_2-z_1)} - \frac{\big( \mu \cdot w + (1-\mu) \cdot \bar{w} \big)\T \cdot (z_3-z_2)}{\big( \mu \cdot c + (1-\mu) \cdot \bar{c} \big)\T \cdot (z_3-z_2)} \DOT
\]
The denominators of both fractions are linear in~$\mu$ and, since the first part of the definition of event~$D$ does not hold, the signs for $\mu = 0$ and $\mu = 1$ are the same and different from~$0$. Hence, both denominators are different from~$0$ for all $\mu \in [0, 1]$. Consequently, function~$f$ is continuous (on $[0, 1]$). As we have
\[
	f(0)
	= \frac{\bar{w}\T \cdot (z_2-z_1)}{\bar{c}\T \cdot (z_2-z_1)} - \frac{\bar{w}\T \cdot (z_3-z_2)}{\bar{c}\T \cdot (z_3-z_2)}
\]
and
\[
	f(1)
	= \frac{w\T \cdot (z_2-z_1)}{c\T \cdot (z_2-z_1)} - \frac{w\T \cdot (z_3-z_2)}{c\T \cdot (z_3-z_2)}
\]
and these differences have different signs as the second part of the definition of event~$D$ is fulfilled, there must be a value $\mu \in [0, 1]$ for which $f(\mu) = 0$. This implies
\[
	\frac{\hat{w}\T \cdot (z_2-z_1)}{\hat{c}\T \cdot (z_2-z_1)} = \frac{\hat{w}\T \cdot (z_3-z_2)}{\hat{c}\T \cdot (z_3-z_2)}
\]
for $\hat{c} \DEF \mu \cdot c + (1-\mu) \cdot \bar{c} \in U_\eps(c)$, $\hat{\lambda} \DEF \mu \cdot \lambda + (1-\mu) \cdot \bar{\lambda} \in U_\eps(\lambda)$, and $\hat{w} \DEF -[u_1, \ldots, u_n] \cdot \hat{\lambda} = \mu \cdot w + (1-\mu) \cdot \bar{w}$. Thus, the event described in Lemma~\ref{lemma:same slopes} occurs.

By applying Lemma~\ref{lemma:swap c} and Lemma~\ref{lemma:same slopes} we obtain
\begin{align*}
	\Pr{D}
	&\leq 2m^n n^{3/2} \eps \phi + \frac{12m^n n^2 \eps^{1/3} \phi}{\delta}
	\leq \frac{4m^n n^2 \eps^{1/3} \phi}{\delta} + \frac{12m^n n^2 \eps^{1/3} \phi}{\delta} \cr
	&= \frac{16m^n n^2 \phi}{\delta} \cdot \eps^{1/3}
	\leq \frac{1}{m^n} \DOT
\end{align*}
This completes the proof.
\end{proof}

Lemma~\ref{lemma:enough random bits} states that if we draw $2n \cdot B(m, n, \phi, \delta)$ random bits for the $2n$ components of~$c$ and~$\lambda$, then the expected number of pivots does not increase significantly. 
We consider now the case that the parameter~$\delta$ is not known (and also no good lower bound). We will use the fraction $\hat{\delta} = \hat{\delta}(n, \phi) \DEF 2n^{3/2}/\phi$ as an estimate for~$\delta$. For the case $\phi > 2n^{3/2}/\delta$, in which the repeated shadow vertex algorithm is guaranteed to yield the optimal solution, this is a valid lower bound for~$\delta$. For the case $\phi < 2n^{3/2}/\delta$ this estimate is too large and we would draw too few random bits, leading to a (for our analysis) unpredictable running time behavior of the shadow vertex method. To solve this problem, we stop the shadow vertex method after at most $8n \cdot p(m, n, \phi, \hat{\delta}(n, \phi))$ pivots, where $p(m, n, \phi, \delta) = O \big( \frac{mn^2}{\delta^2} + \frac{m \sqrt{n} \phi}{\delta} \big)$ is the upper bound for the expected number of pivots stated in Lemma~\ref{lemma:enough random bits}. When the shadow vertex method stops, we assume that the current choice of~$\phi$ is too small (although this does not have to be the case) and restart the repeated shadow vertex algorithm with $2\phi$. Recall that this is the same doubling strategey that is applied when the repeated shadow vertex algorithm yields a non-optimal solution for the original linear program. We call this algorithm the shadow vertex algorithm with random bits.

\begin{theorem}
The shadow vertex algorithm with random bits solves linear programs with~$n$ variables and~$m$ constraints satisfying the $\delta$-distance property using $O \big ( \frac{mn^4}{\delta^2} \cdot \log \big( \frac{1}{\delta} \big) \big)$ pivots in expectation if a feasible solution is given.
\end{theorem}

Note that, in analogy, all other results stated in Theorem~\ref{thm:MainNumberOfPivots} and Theorem~\ref{thm:MainNumberOfPivots2} also hold for the shadow vertex algorithm with random bits with an additional $O(n)$-factor (or $O(m)$-factor when no feasible solution is given).

\begin{proof}
Let us assume that the shadow vertex algorithm with random bits does not find the optimal solution before the first iteration~$i^\star$ for which $\phi_{i^\star} > 2n^{3/2}/\delta$. For iterations $i \geq i^\star$ we know that the shadow vertex algorithm will return the optimal solution (or detect, that the linear program is unbounded) if it is not stopped because the number of pivots exceeds $8n \cdot p(m, n, \phi_i, \hat{\delta}(n, \phi_i))$. Due to Markov's inequality, the probability of the latter event is bounded from above by $1/8n$ (for each facet of the optimal solution) because $p(m, n, \phi_i, \hat{\delta}(n, \phi_i)) \geq p(m, n, \phi_i, \delta)$ due to $\hat{\delta}(n, \phi_i) \leq \delta$ and $p(m, n, \phi_i, \delta)$ is an upper bound for the expected number of pivots. As~$n$ facets have to be identified in iteration~$i$, the probability that the shadow vertex method stops because of too many pivots is bounded from above by $n \cdot 1/8n = 1/8$. Hence, the expected number of pivots of all iterations $i \geq i^\star$, provided that iteration~$i^\star$ is reached, is at most
\begin{align*}
	&\sum_{i=i^\star}^\infty \left( \frac{1}{8} \right)^{i-i^\star} \cdot \frac{7}{8} \cdot n \cdot 8n \cdot p(m, n, \phi_i, \hat{\delta}(n, \phi_i)) \cr
	= &7n^2 \cdot \sum_{i=i^\star}^\infty \frac{1}{8^{i-i^\star}} \cdot p \left( m, n, \phi_i, \frac{2n^{3/2}}{\phi_i} \right) \cr
	= &O \left( 8^{i^\star} n^2 \cdot \sum_{i=i^\star}^\infty \frac{1}{8^i} \cdot \frac{m \sqrt{n} \phi_i}{\frac{2n^{3/2}}{\phi_i}} \right)
	= O \left( 8^{i^\star} n \cdot \sum_{i=i^\star}^\infty \frac{1}{8^i} \cdot m \phi_i^2 \right) \cr
	= &O \left( 8^{i^\star} n \cdot \sum_{i=i^\star}^\infty \frac{1}{8^i} \cdot m \cdot (2^i n^{3/2})^2 \right)
	= O \left( 8^{i^\star} n \cdot \sum_{i=i^\star}^\infty \frac{1}{2^i} \cdot m n^3 \right) \cr
	= &O(4^{i^\star}m n^4)
	= O \left( \frac{m n^4}{\delta^2} \right) \DOT
\end{align*}
Some equations require further explanation. The factor $n \cdot 8n \cdot p(m, n, \phi_i, \hat{\delta}(n, \phi_i))$ stems from the fact that we have to identify~$n$ facets, and for each we stop after at most $8n \cdot p(m, n, \phi_i, \hat{\delta}(n, \phi_i))$ pivots. The second equation is in accordance with Lemma~\ref{lemma:enough random bits}, which states that $p(m, n, \phi, \delta) = O \big( \frac{mn^2}{\delta^2} + \frac{m \sqrt{n} \phi}{\delta} \big)$. As the term $mn^2/\delta^2$ is dominated by the term $m \sqrt{n} \phi/\delta$ when $\phi \geq n^{3/2}/\delta$, it can be omitted in the $O$-notation for such values. Above we only consider iterations $i \geq i^\star$, i.e., $\phi_i \geq \phi_{i^\star} > 2n^{3/2}/\delta$. The last equation is due to the fact that
\[
	2^{i^\star-1} n^{3/2}
	= \phi_{i^\star-1}
	\leq \frac{2n^{3/2}}{\delta} \COMMA
\]
i.e., $2^{i^\star} \leq 4/\delta$ and, hence, $4^{i^\star} = O(1/\delta^2)$.

To finish the proof, we observe that the iterations $i = 1, \ldots, i^\star$ require at most
\begin{align*}
	&\sum_{i=1}^{i^\star-1} n \cdot 8n \cdot p(m, n, \phi_i, \hat{\delta}(n, \phi))
	= \sum_{i=1}^{i^\star-1} n \cdot 8n \cdot p \left( m, n, \phi_i, \frac{2n^{3/2}}{\phi_i} \right) \cr
	= &O \left( \sum_{i=1}^{i^\star-1} n^2 \cdot \frac{mn^2}{\delta^2} \right)
	= O \left( i^\star \cdot \frac{mn^4}{\delta^2} \right)
	= O \left( \log \left( \frac{1}{\delta} \right) \cdot \frac{mn^4}{\delta^2} \right)
\end{align*}
pivots in expectation. The second equation stems from Lemma~\ref{lemma:enough random bits}, which states that $p(m, n, \phi, \delta) = O \big( \frac{mn^2}{\delta^2} + \frac{m \sqrt{n} \phi}{\delta} \big)$. The second term in the sum can be omitted if $\phi = O(n^{3/2}/\delta)$, which is the case for $\phi_1, \ldots, \phi_{i^\star-1}$. Finally, $i^\star$ is the smallest integer~$i$ for which $2^i n^{3/2} > 2n^{3/2}/\delta$. Hence, $i^\star = O(\log(1/\delta))$.
\end{proof}

\end{appendix}

\end{document}